\documentclass[sigconf]{acmart}
\AtBeginDocument{%
  }
\newcommand{\mycomment}[1]{}

\setcopyright{none}                 
\acmConference[]{}{}{}             
\acmBooktitle{} \acmPrice{} \acmDOI{} \acmISBN{}
\renewcommand\footnotetextcopyrightpermission[1]{}

\usepackage{algorithm}
\usepackage{algpseudocode}
\usepackage{fancyvrb}

\usepackage{amsmath} 
\usepackage{booktabs}

\usepackage{multirow}   
\usepackage[most]{tcolorbox}

\newtheorem{theorem}{Theorem}[section]

\newtheorem{Assumption}{Assumption}
\newtheorem{Definition}{Definition}

\usepackage{graphicx}
\graphicspath{{Figures/}} 
\usepackage{subcaption}

\usepackage{xcolor}
\usepackage{tcolorbox}
\usepackage{listings}

\newcommand{\ASImplementation}[1]{}

\definecolor{lightgray}{gray}{0.95}

\usepackage{enumitem}
\newlist{casesp}{enumerate}{3} 
\setlist[casesp]{align=left, 
                 listparindent=\parindent, 
                 parsep=\parskip, 
                 font=\normalfont\bfseries, 
                 leftmargin=0pt, 
                 labelwidth=0pt, 
                 itemindent=.4em,labelsep=.4em, 
                 partopsep=0pt, 
                 }
\setlist[casesp,1]{label=Case~\arabic*:,ref=\arabic*,leftmargin=0pt}
\setlist[casesp,2]{label=Case~\thecasespi.\roman*:,ref=\thecasespi.\roman*,leftmargin=10pt}
\setlist[casesp,3]{label=Case~\thecasespii.\alph*:,ref=\thecasespii.\alph*,leftmargin=20pt}

\setcopyright{acmlicensed}
\copyrightyear{2018}
\acmYear{2018}
\acmDOI{XXXXXXX.XXXXXXX}
\acmConference[Conference acronym 'XX]{Make sure to enter the correct
  conference title from your rights confirmation email}{June 03--05,
  2018}{Woodstock, NY}
\acmISBN{978-1-4503-XXXX-X/2018/06}

\begin{document}
\title{LARK - Linearizability Algorithms for Replicated Keys in Aerospike}


\author{Andrew Gooding}
\orcid{0009-0009-5848-8356}
\authornote{This work was performed at Aerospike.}
\affiliation{%
  \institution{Consultant}
  \streetaddress{}
  \city{Mountain View}
  \state{California}
  \postcode{94040}
}
\email{gooding470@hotmail.com}

\author{Kevin Porter}
\orcid{0009-0004-3698-3923}
\affiliation{%
  \institution{Aerospike}
  \streetaddress{}
  \city{Mountain View}
  \state{California}
  \postcode{94040}
}
\email{kporter@aerospike.com}

\author{Thomas Lopatic}
\orcid{0009-0007-7952-6599}
\authornotemark[1]
\affiliation{%
  \institution{Consultant}
  \streetaddress{}
  \city{Berlin}
  \state{Germany}
  \postcode{}
}
\email{thomas@lopatic.de}

\author{Ashish Krishnadeo Shinde}
\orcid{0009-0006-5797-3745}
\authornotemark[1]
\affiliation{%
  \institution{Divyam AI}
  \city{Bengaluru}
  \country{India}
}
\email{omkarashish@gmail.com}

\author{Sunil Sayyaparaju}
\orcid{0009-0003-1842-6739}
\affiliation{%
  \institution{Aerospike}
  \city{Bengaluru}
  \country{India}
}
\email{sunil@aerospike.com}

\author{Srinivasan Seshadri}
\orcid{0009-0000-1450-1111}
\affiliation{%
  \institution{Aerospike}
  \streetaddress{}
  \city{Mountain View}
  \state{California}
  \postcode{}
}
\email{sseshadri@aerospike.com}

\author{V. Srinivasan}
\orcid{0009-0007-1336-3555}
\affiliation{%
  \institution{Aerospike}
  \streetaddress{}
  \city{Mountain View}
  \state{California}
  \postcode{}
}
\email{srini@aerospike.com}

\renewcommand{\shortauthors}{Gooding et. al.}


\begin{abstract}
We present \textbf{LARK} (Linearizability Algorithms for Replicated Keys), a synchronous replication protocol that achieves linearizability while minimizing latency and infrastructure cost, at significantly higher availability than traditional quorum-log consensus. LARK introduces \emph{Partition Availability Conditions (PAC)} that reason over the entire database cluster rather than fixed replica sets, improving partition availability under independent failures by roughly \(3\times\) when tolerating one failure and \(10\times\) when tolerating two. Unlike Raft, Paxos, and Viewstamped Replication, LARK eliminates ordered logs, enabling immediate partition readiness after leader changes—with at most a per-key duplicate-resolution round trip when the new leader lacks the latest copy. Under equal storage budgets—where both systems maintain only \(f{+}1\) data copies to tolerate \(f\) failures—LARK continues committing through data-node failures while log-based protocols must pause commits for replica rebuilding. These properties also enable zero-downtime rolling restarts even when maintaining only two copies. We provide formal safety arguments and a TLA+ specification, and we demonstrate through analysis and experiments that LARK achieves significant availability gains.
\end{abstract}



\maketitle

\section{Introduction}
\label{sec:intro}

Distributed databases increasingly serve latency-sensitive applications that demand both \emph{high availability} and \emph{strong consistency}, even in the presence of failures. Examples include online advertising, gaming, financial services, and personalization systems where even brief stalls can degrade user experience, impact revenue, or violate strict service-level objectives (SLOs). These systems often operate at sub-millisecond latency targets and must minimize downtime during both planned and unplanned events.  

Achieving linearizable reads and writes at this scale is traditionally done via quorum-log consensus protocols such as Paxos~\cite{lamport2001paxos,lamport1998tpp,vanrenesse2015moderately}\footnote{In this paper, “Paxos” refers to the log-backed SMR variant (\emph{Multi-Paxos}).}
, Raft~\cite{ongaro2014search}, or Viewstamped Replication (VR)~\cite{liskov12vr}. These approaches elect a leader per partition (or shard) and replicate writes to a quorum of replicas via an \emph{ordered log}. However, these designs impose well-known costs at scale:  
\begin{itemize}
    \item \textbf{Availability limitations:} A partition becomes unavailable if fewer than $f{+}1$ of its $2f{+}1$ configured replicas are reachable.  
    \item \textbf{Transition delays:} Leader changes require log catch-up (prefix reconciliation or snapshot replay), temporarily stalling the partition.  
    \item \textbf{Operational complexity:} Ordered logs introduce write amplification, replay overhead, and storage compaction challenges.
\end{itemize}

These challenges are particularly acute in cost-sensitive deployments that minimize replication factors (\(RF\)) to control infrastructure costs while relying on fast but expensive storage like NVMe SSDs. Reducing \(RF\) is desirable but worsens unavailability under quorum-log protocols. \textbf{LARK} (Linearizability Algorithms for Replicated Keys) addresses this tension directly.  

\subsection*{Introducing LARK}

\textbf{LARK} (Linearizability Algorithms for Replicated Keys) is the synchronous replication design in the Aerospike database~\cite{aerospike:book2024, as:vldb2016, as:vldb2023, as:vldb2012, as:sigmod2025}. Based on deployment requirements of our customers, the primary design goals of LARK are to provide linearizability with minimal latency and infrastructure cost while maximizing availability. 
Therefore, LARK replaces per-partition quorum logs with \emph{Partition Availability Conditions (PAC)} and a \emph{log-free} state-replication path. PAC broadens availability beyond replica-set majority by reasoning over the \emph{database-wide cluster} and significantly expands the conditions under which partitions remain safely available. LARK removes ordered logs entirely. Writes are applied directly to the key-value store, and correctness is ensured via \emph{logical clocks} and per-key \emph{duplicate-resolution checks}. After leader changes, keys for which the new leader holds the \emph{latest committed copy} become \emph{immediately ready}; others complete a short duplicate-resolution round trip, avoiding log catch-up. Reads never depend on log indices or replay, simplifying the steady-state path. We store exactly $RF\!=\!f{+}1$ copies (to tolerate $f$ failures), and \emph{re-replication (migration)} of a replacement replica runs asynchronously in the background, so commit progress never hinges on bringing a spare to log parity.

\paragraph{Increased Availability from PAC}
PAC introduces four conditions under which a partition stays available (Section~\ref{sec:partitionavailability}). For example, a \emph{simple-majority} condition that makes a partition available whenever a majority of database nodes are up and at least one \emph{full replica} (i.e., a replica holding the latest committed copy of all records in the partition) is reachable. Under independent node failures, the \emph{simple-majority} condition alone  delivers significant gains. Our analysis and simulations (Section~\ref{sec:expts}) show that LARK improves partition availability by roughly $3\times$ at $RF = 2$ and $10\times$ at $RF = 3$ compared to quorum-log systems. 
PAC also includes a \emph{super-majority} condition that enables zero-downtime rolling restarts even at $RF{=}2$. Because PAC is independent of any single fixed replica set, partitions are not stranded simply because some preconfigured members are temporarily missing.

\paragraph{Equal storage: commits without log catch-up.}
A second, distinct availability benefit arises under an \emph{equal storage budget}, where both systems maintain only $f{+}1$  data copies\footnote{To ensure all schemes have equal storage, we assume quorum-log schemes only have $f{+}1$ log-persisting data replicas (with up to $2f{+}1$ voters overall), following common cost-reduction patterns in Paxos-family deployments~\cite{lamport2004cheap}.}. In quorum-log designs, even with $2f{+}1$ voters, losing one data replica leaves only $f$ log-persisting voters; the leader cannot commit new entries until a spare voter is \emph{caught up on the log} (typically via snapshot/state-transfer plus backfill), creating a no-commit window~\cite{ongaro2014raft}. LARK, in contrast, continues committing new writes immediately while the replacement copy re-replicates in the background.
In time-series microbenchmarks with a 5-minute outage (Section~\ref{sec:expts}), LARK sustains service throughout while the baseline pauses for roughly \texttt{partition\_size}/\texttt{network\_bandwidth} seconds; when both serve, latencies are comparable.

\paragraph{Zero-downtime rolling restarts with $RF{=}2$.}
Under \emph{SuperMajority} (fewer than $RF$ nodes unavailable), rolling restarts proceed with no downtime at $RF{=}2$: when one original replica reboots, the other serves with an \emph{interim} second copy; upon return they swap roles; when both originals are back, the interim retires. The interim accepts only new updates (no historical backfill), so when originals return only accrued deltas flow. In quorum-log systems, an interim must first catch up (log and/or snapshot) before accepting writes, extending the maintenance window~\cite{ongaro2014raft}.

\paragraph{Write continuity during leadership changes.}
Because LARK tolerates a bounded view skew (at most one regime) between nodes, many in-flight operations around a leader change complete without client retry. The correctness argument later uses this explicit bound.

\paragraph{RSM scope.}
LARK implements a replicated state machine at \emph{per-key} granularity rather than maintaining a per-partition ordered log. This scope matches our target workloads, which require linearizability of individual key values and benefit from immediate leader readiness and low operational overhead.

\paragraph{Record-size limitation.}
The efficiency of record writes when one of the replicas does not have the latest copy relies on the size of the record as opposed to the size of the update to the record\footnote{Aerospike limits records to a maximum of 8MB in size currently.}. 

LARK has been deployed for years in production \textsc{Aerospike} clusters, validating correctness and operational benefits at scale. Aerospike’s strong-consistency mode was independently evaluated by Jepsen in 2018~\cite{jepsen-aerospike-39903}; Aerospike subsequently described fixes in version~4.0~\cite{aerospike-jepsen-40}. Enabling linearizable reads under LARK adds only modest overhead—often about one additional intra-cluster RTT on common paths—relative to eventual-consistency mode~\cite{aerospike-jepsen-40}. This paper formalizes and generalizes the synchronous-replication design we call LARK and documents improvements made since that evaluation.

The rest of the paper is organized as follows:
Section~\ref{sec:aerospike-architecture} describes the system model and definitions. Section~\ref{sec:partitionavailability} presents (PAC) partition availability conditions. Section~\ref{sec:lark-algorithm} details the LARK algorithm. Section~\ref{sec:expts} reports the experimental results. Section~\ref{sec:related} reviews related work. Section~\ref{sec:conclusion} concludes the paper. Appendices provide proofs, additional experiments, and auxiliary analysis.

\section{System Model and Definitions}
\label{sec:aerospike-architecture}
\begin{figure}
\centering
\includegraphics[width=1.0\linewidth]{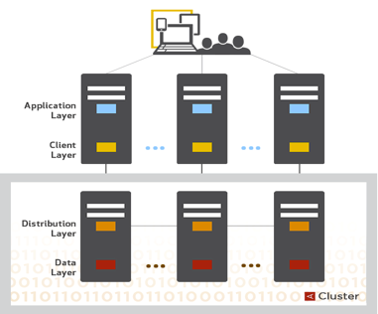}
\caption{Aerospike Cluster Architecture}
\label{fig:as_arch}
\end{figure}

The Aerospike real-time database cluster~\cite{as:vldb2023} is the operational substrate for LARK (Figure~\ref{fig:as_arch}). We highlight four properties relevant to this paper:
\begin{itemize}
  \item \textbf{Shared-nothing nodes:} all nodes are identical peers with local storage.
  \item \textbf{Namespaces:} records live in \emph{namespaces}\footnote{Namespaces resemble tablespaces; within a namespace, \emph{sets} are analogous to tables.}. Unless noted, discussion refers to records within a single namespace.
  \item \textbf{Uniform partitioning:} keys are mapped to a fixed number of partitions, preventing hotspots.
  \item \textbf{One-hop clients:} intelligent clients cache the partition$\rightarrow$leader mapping and route requests directly.
\end{itemize}

\begin{figure}[h]
\centering
\includegraphics[width=1.0\linewidth, height=3in]{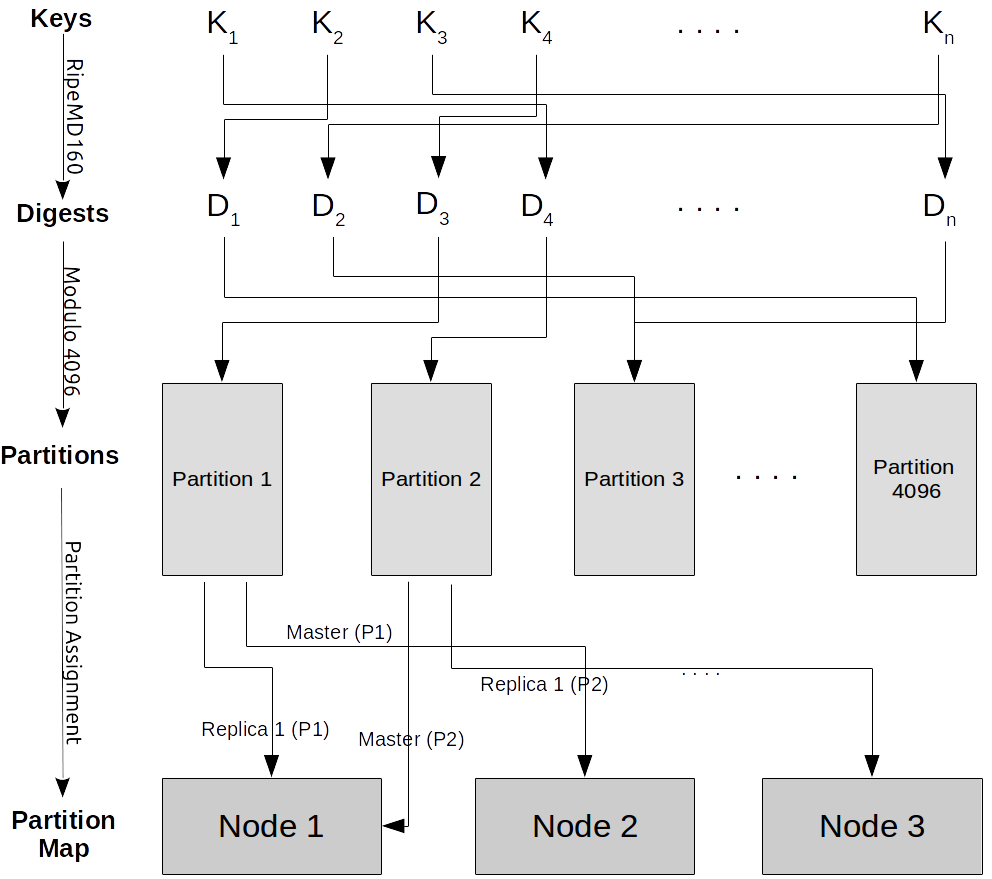}
\caption{Data Partitioning}
\label{fig:data_partitioning}
\end{figure}

\subsection{Data Partitioning and Placement}

Aerospike distributes data uniformly across nodes (Figure~\ref{fig:data_partitioning}). A record’s primary key is hashed to a 160-bit digest using RIPEMD-160~\cite{infosec:ripemd}. The digest space is partitioned into 4096 non-overlapping partitions, which are the unit of placement. Records are assigned to partitions by hashing their primary keys; even with skewed key distributions, the induced distribution over digests—and thus partitions—is uniform.

\begin{figure}
\centering
\includegraphics[width=1.0\linewidth, height=2in]{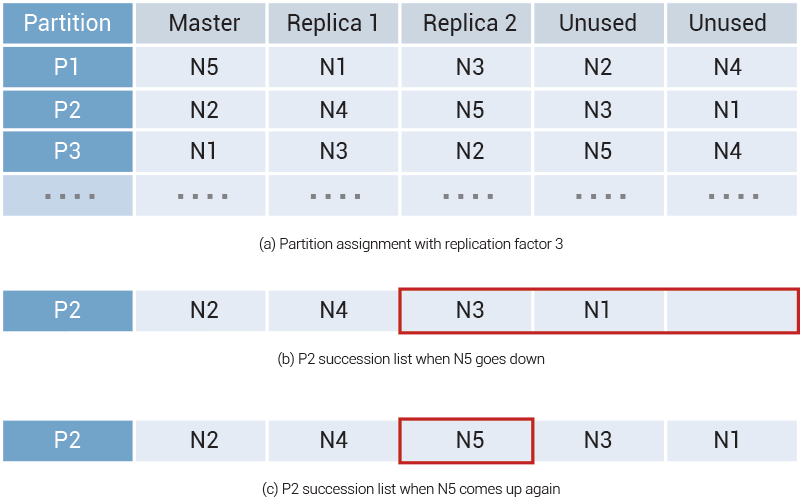}
\caption{Mapping Partitions to Nodes}
\label{fig:succession_list}
\end{figure}

Let $RF$ be the replication factor ($RF = f{+}1$ tolerates $f$ failures). Partitions select their \emph{roster replicas} via Rendezvous hashing~\cite{tr:rendezvous}:
\begin{enumerate}
 \item For each partition $P$ and node $N$, compute a score on $(P,N)$ using a hash.\footnote{Any collision-resistant hash suffices.}
  \item Sort nodes by score to obtain $P$’s \emph{succession list}.
  \item The first $RF$ nodes form the \emph{roster replicas}; the first is the \emph{roster leader} (displayed as Master in Figures~\ref{fig:data_partitioning} and~\ref{fig:succession_list}). When we need not distinguish, we refer to leader and followers collectively as \emph{roster replicas}.
\end{enumerate}
Figure~\ref{fig:succession_list}(a) shows the assignment for a 5-node cluster with $RF{=}3$. LARK is placement-agnostic: any deterministic scheme that yields a per-partition succession list is acceptable.

\subsection{Clustering}
\label{subsec:reclustering}

For expository clarity, we first assume a fixed roster and analyze node up/down dynamics; roster changes are handled in Section~\ref{subsec:rosterchanges}.

Nodes exchange periodic heartbeats to maintain membership. On connectivity changes, a \emph{reclustering} step identifies disjoint \emph{clusters}, each a maximal set of nodes with full mutual reachability. Nodes in a cluster run one consensus round to agree on \texttt{ClusterMembers} (the cluster view) and mint a monotonically increasing \emph{exchange number}. All nodes in the same cluster adopt the same exchange number.

Given \texttt{ClusterMembers}, each node independently computes a partition’s \emph{cluster replicas}: the first $RF$ nodes in the partition’s succession list that are present in the cluster. A deterministic tie-break then selects the \emph{cluster leader} as described in Section~\ref{subsec:rebalance}.

Figure~\ref{fig:succession_list} illustrates a common scenario. When node N5 fails, it is removed from the succession lists (e.g., for partition P2), causing a left shift so that N3 assumes N5’s roster position and P2’s records are migrated to N3 (Figure~\ref{fig:succession_list}(b)). When N5 returns, it regains its position (Figure~\ref{fig:succession_list}(c)). If a partition had no replica on N5 (e.g., P3), no migration is required. Adding a brand-new node inserts it into each succession list, right-shifting lower-ranked nodes; assignments to the left remain unchanged.

\section{Partition Availability Conditions (PAC)}
\label{sec:partitionavailability}

LARK declares a partition \emph{available} after each reclustering step if at least one of a small set of cluster-scoped predicates holds. These \emph{Partition Availability Conditions (PAC)} reason over the \emph{database-wide} cluster rather than a fixed replica set, which is the source of LARK's availability advantage.
\paragraph{Full replica.}
A node is \emph{full} for partition $P$ if it holds the latest committed version of every record in $P$.

A partition $P$ is \textbf{available in a cluster} if \emph{any one} of the following holds:
\begin{enumerate}[leftmargin=1.2em]
  \item \textbf{SuperMajority:} The cluster contains a strict majority of roster nodes and fewer than $RF$ roster nodes are missing. (Hence at least one roster replica is present.)
  \item \textbf{AllRosterReplicas:} All $RF$ roster replicas of $P$ are present in the cluster. 
  \item \textbf{SimpleMajority:} The cluster contains a strict majority of roster nodes, includes at least one roster replica of $P$, and at least one node is \emph{full} for $P$.
  \item \textbf{HalfRoster:} Exactly half of the roster nodes are present, the roster leader of $P$ is present, and at least one node is \emph{full} for $P$.
\end{enumerate}
\paragraph{Regime.}
Each reclustering step assigns the cluster a monotonically increasing \emph{exchange number}. For partition $P$, we refer to the exchange number in effect when it is available and serving requests as $P$’s \emph{regime number (PR)}.

\subsection{PAC Safety}
\label{subsec:pac-correctness}

Safety is proved by establishing two properties PAC must guarantee for each partition $P$:

\begin{enumerate}[leftmargin=1.2em]
  \item \textbf{Leader uniqueness.} At any time, at most one cluster in the system that satisfies PAC for $P$ can \emph{successfully} serve reads and writes.
  \item \textbf{Access to latest state.} The (unique) serving leader must have access to the latest committed version of every record in $P$.
\end{enumerate}

These are proved through a sequence of lemmas (proofs in Appendix~\ref{app:proof}).

\begin{lemma}
\label{lemma:mainonerosterreplica}
Any cluster that satisfies one of the PAC rules for a given partition must include at least one roster replica of that partition.
\end{lemma}

\begin{lemma}
\label{lemma:mainonecommonnode}
Let $C_1$ and $C_2$ be two distinct clusters that both satisfy PAC for a partition. Then $C_1$ and $C_2$ must share at least one node.
\end{lemma}

\begin{lemma}
\label{lemma:mainonlyonecluster}
During any regime, there is at most one cluster in the system that satisfies PAC for a given partition.
\end{lemma}

\begin{lemma}
\label{lemma:mainclusterreplicacommon}
Let $C_1$ and $C_2$ be two clusters available for partition $P$, with regime numbers $R_1$ and $R_2$ such that $R_1 < R_2$ and no intermediate regime exists where $P$ was available. Then at least one of the cluster replicas from $C_1$ is also present in $C_2$.
\end{lemma}

\section{LARK Algorithm}
\label{sec:lark-algorithm}

\begin{table*}[t]
\small
\centering
\caption{Notation and metadata used by algorithms and proofs}
\label{tab:notation}
\begin{tabular}{@{}llp{0.62\textwidth}@{}}
\toprule
\textbf{Notation} & \textbf{Scope} & \textbf{Meaning} \\
\midrule
\multicolumn{3}{@{}l}{\textit{Cluster identifiers}} \\
\addlinespace[2pt]
$RF$ & cluster & Replication factor ($f{+}1$). \\
$P$ & partition & Partition identifier. \\
\midrule
\multicolumn{3}{@{}l}{\textit{Global/partition clocks}} \\
\addlinespace[2pt]
$ER$ & node & \textit{Exchange number} (cluster epoch minted by reclustering). \\
$PR$ & node$\times$partition & \textit{Partition regime}; set to $ER$ when the partition becomes available. \\
$LR$ & node$\times$partition & \textit{Leader regime}: $PR$ at which the current leader was first elected. \\
\midrule
\multicolumn{3}{@{}l}{\textit{Per-record metadata (stored with each version)}} \\
\addlinespace[2pt]
\texttt{Key.RR} & record version & Record’s regime tag: the $PR$ in effect when this version was (re-)replicated. \\
$VN$ & record version & Version number within a given $RR$ (used implicitly). \\
$LC$ & record version & Logical clock used for per-key ordering and dup-res; \emph{defined as the lexicographic pair} $(RR, VN)$. For brevity, we refer only to $LC$ elsewhere in the paper. \\
status & record version & $\{\textit{replicated}, \textit{unreplicated}\}$. \\
\midrule
\multicolumn{3}{@{}l}{\textit{Leader/cluster state}} \\
\addlinespace[2pt]
\texttt{NodesInCluster} & node & Local view of current cluster members (used for decisions). \\
\textit{full} & node$\times$partition & Node has the latest version of every record in $P$. \\
\textit{duplicate} & node$\times$partition & Node may hold the latest version of some record in $P$. \\
\midrule
\multicolumn{3}{@{}l}{\textit{Message fields, helpers, and rules}} \\
\addlinespace[2pt]
$LRM$ & message field & Leader’s $LR$ piggy-backed on \textsc{Replica-Write}. \\
$\textsc{Replicas}(\texttt{NodesInCluster},P)$ & helper & First $RF$ nodes of $P$’s succession list that are in \texttt{NodesInCluster}. \\
\texttt{check\_regime}($N,PR$) & helper & Success iff $N$’s $PR{=}PR$ and $N$ recognizes the caller as leader for $P$. \\
PR-Match for Migration & rule & Migrate into leader only when sender and leader share the same $PR$. \\
\bottomrule
\end{tabular}
\end{table*}

Table~\ref{tab:notation} contains a glossary of terms and their intuitive meanings we will use in the rest of the paper for ready reference.

We will now describe the details of LARK. LARK consists of the following algorithms which we will describe one after the other:
\begin{enumerate}
\item A Scalable Global Clustering Algorithm
\item Rebalancing of data amongst cluster replicas of a partition after reclustering
\item Reads and Writes
\end{enumerate}

\subsection{Global Reclustering Algorithm}
\label{subsec:globalclustering}

LARK’s \emph{Partition Availability Conditions (PAC)} reason over the \emph{database-wide} cluster. In particular, they depend on how many nodes of the roster are in the cluster, which is the key insight that provides LARK its availability advantage over Raft/VR as shown in Section~\ref{sec:expts}. Therefore, LARK must maintain an authoritative, agreed-upon \emph{global cluster membership}. 

When cluster membership changes (node joins, departures, or failures), reclustering performs three steps:
\begin{enumerate}
    \item \textbf{Consensus on ClusterMembers:} Nodes continuously exchange heartbeats over direct links. For a cluster of size \(n\), this involves approximately \(n(n-1)/2\) peer connections per period. When a connectivity change stabilizes, nodes run a single consensus step to finalize the new \texttt{ClusterMembers}.
    \item \textbf{Minting a new exchange number:} A unique, monotonically increasing \emph{exchange} number is allocated to each node in the cluster. 
    \item \textbf{Deterministic cluster replica/leader computation:} Given \texttt{ClusterMembers}, each node independently computes cluster replicas as the first $RF$ nodes in the succession list that are also in the cluster.  
\end{enumerate}

Once these steps complete, each node atomically updates its \texttt{ClusterMembers}, \texttt{exchange number}, and local \emph{succession lists} derived from the roster. All nodes now agree on the same cluster view, enabling PAC-based decisions to proceed safely.

\paragraph{Why global reclustering is not a scalability bottleneck.}
At first glance, a global step sounds costly; in practice, steady-state control traffic is dominated by heartbeats, and the one-shot consensus to mint a new exchange number is linear.

\emph{LARK.} Nodes maintain full-mesh heartbeats: \(O(n(n-1))\) tiny messages per period for a cluster of \(n\) nodes. When a connectivity change stabilizes, reclustering adds a single consensus round to finalize \texttt{ClusterMembers} and mint \(ER\), which is \(O(n)\).

\emph{Quorum-log protocols (Raft/VR).} With \(P\) partitions and replication factor \(RF\), per-partition leaders send heartbeats to replicas each period: \(O\!\left(P \cdot RF \cdot (RF-1)\right)\). Transient membership changes or leader failures can also trigger per-partition elections of similar order. Thus, for typical deployments where \(RF\) is a small constant (e.g., \(RF{=}3\)), the control-plane message rate scales with \(P\), not \(n\). It turns out these two seem to be the same around $n=157$\footnote{Comparing steady-state heartbeats,
\[
n(n-1) \;\approx\; P \cdot RF \cdot (RF-1)
\quad\Rightarrow\quad
n \;\approx\; \sqrt{P \cdot RF \cdot (RF-1)}.
\]
For \(RF{=}3\), this is \(n \!\approx\! \sqrt{6P}\). With \(P{=}4096\) (Aerospike’s default), \(\sqrt{6P}\!=\!\sqrt{24576}\!\approx\!156.8\). So LARK’s full-mesh heartbeat volume matches a Raft/VR deployment with \(RF{=}3\) at around \(n\!\approx\!157\) nodes.}
Empirically, LARK operates comfortably from tens to low-hundreds of nodes per cluster; this already covers multi-petabyte clusters with modern nodes offering $\sim$100\,TB of local storage each. As described in Section~\ref{sec:conclusion}, we plan to extend LARK’s applicability to thousands of nodes.

\subsection{Rebalancing of a Partition}
\label{subsec:rebalance}

After a cluster is formed via reclustering (Section~\ref{subsec:globalclustering}), each node independently performs a local \textit{rebalance} operation for every partition it may be responsible for as part of this new cluster. The purpose of rebalance is to determine:
\begin{itemize}
    \item whether a partition is available under the current PAC rules,
    \item which nodes become the new cluster replicas, and
    \item which node becomes the cluster leader.
\end{itemize}

\begin{sloppypar}
    Rebalance is triggered only after the clustering subsystem has atomically updated the node’s exchange number and the \texttt{ClusterMembers} variable. Any new reclustering during this process will cancel the in-progress rebalance and restart it.
\end{sloppypar}

A few quick definitions before we begin describing the rebalance process:

\begin{itemize}
    \item  [\bf PR:]

Each node maintains a \textit{partition regime} (PR) for each partition it stores, which is used as a logical timestamp indicating the current partition version. When a partition becomes available within a new cluster, all its cluster replicas update their PR to match the node’s exchange number.

\item [\bf LR:]
To track leadership history, the system maintains a \textit{leader regime} (LR) for each partition, which records the PR at which the current leader was first elected. This is used later to decide what delayed writes if any to accept. 
\end{itemize}
The rebalance process consists of the following steps:

\begin{enumerate}

    \item \textbf{Exchange Full Status:} Each node predicts whether it will be \textit{full} after rebalance. A node is considered full if:
    \begin{itemize}
        \item its current PR is one less than the new exchange number, and
        \item it is full in the current PR.
    \end{itemize}

    \item \textbf{Evaluate Partition Availability:} Each node independently evaluates PAC based on:
    \begin{itemize}
        \item the current cluster membership,
        \item the succession list of the partition, and
        \item the predicted full status of nodes.
    \end{itemize}
    If the partition is not available, rebalance terminates and the node marks itself as not full for the partition. The remaining steps are skipped.

    \item \textbf{Retain Previous Leader (if applicable):} If the current leader is in \texttt{ClusterMembers} and is a cluster replica, it remains the leader for the new regime. The leader shares its LR with the rest of the cluster. \ASImplementation{In our implementation we always eventually move leadership to the leftmost node. This has the advantage that load balance is restored after failures. We chose not to force one more migration in this paper to keep the discussion simpler - moving leadership to leftmost node along with the leader not changed optimization (which is not present in our implementation) still works.}

    \item \textbf{Atomically Update Local State:} Each node then updates the following variables atomically:
    \begin{itemize}
        \item Set the new partition regime PR = exchange number.
        \item Mark the full status of the node.
        \item Copy \texttt{ClusterMembers} into a new variable \texttt{NodesInCluster} used for local read/write decisions.
        \item If a leader has been chosen in step 4, update LR using the value provided by that leader.
        \item If a leader was not retained from step 4:
        \begin{itemize}
            \item If there exists a full node from the previous regime, the first full node (by succession list order) becomes the cluster leader. LR is set to the new PR. If the chosen leader is not among the top $RF$ nodes in the succession list, it serves as an \textit{acting leader} and will later transfer leadership to the first cluster replica in the succession list. 
            \item If no node was full, the first available node in the succession list becomes the leader. LR is again set to PR.
        \end{itemize}
    \end{itemize}

    \item \textbf{Leader Immigration (if needed):} If the new cluster leader is not full, it begins migration of the latest versions of records from any nodes (including the acting leader) that may have them (such nodes called duplicates are formally defined in Section~\ref{subsec:duplicates}). This step guarantees eventual freshness.

    \item \textbf{Replica Emigration (if needed):} Once the leader becomes full, it proactively migrates the latest versions to all other cluster replicas, ensuring they also become full.
\end{enumerate}

\subsubsection{Atomicity of Rebalance Steps}

Within the rebalance process, Steps 1 through 3 can proceed concurrently with reads and writes. Step 4, which updates shared variables such as the partition regime, full status, and cluster membership view, must be performed atomically with respect to reads and writes. Note that Step 4 involves minimal local logic only and should take of the order of hundreds of nanoseconds to a few microseconds to complete. 

Migration steps (Steps 5 and 6), if required, are performed asynchronously. To ensure consistency, we introduce the constraint \textbf{PR Match for Migration}: a node may migrate its records into the current cluster leader only if both nodes share the same partition regime (PR). This constraint is essential for correctness, as formalized in Section~\ref{app:proof}.

The full status of cluster leaders and cluster replicas is a shared variable accessed by both rebalance and the read/write path. It is updated atomically upon completion of migrations.

\subsubsection{Duplicates}
\label{subsec:duplicates}

In Step 5 of the rebalance process, we refer to nodes that might hold the latest version of a record in a partition. We call these nodes \textit{duplicates}.

A node $N$ becomes a duplicate for a partition when it becomes a  cluster replica.  $N$ can be removed as a duplicate when: it is part of a cluster in which the partition is available and it is not a cluster replica and the leader has migrated its latest record versions into the cluster replicas (after step 6 of the rebalance process). At this point, the responsibility for holding the latest versions is fully transferred to the new cluster replicas, all of which are now considered duplicates.  We will need the notion of duplicates in Section~\ref{subsec:asyncreadsandwrites}.

\subsection{High Level Overview of Reads and Writes}

Clients always send read and write requests to the current leader of a partition. If the contacted node is not the actual leader, it proxies the request to the correct leader (though this is elided in the algorithms for clarity). Writes are always propagated to all $RF$ cluster replicas. A write is acknowledged to the client only after all replicas have accepted it. 

Each version of a record at a node is associated with a \textit{replication status}, which can be either \textit{replicated} or \textit{unreplicated}. A version can be marked replicated once all $RF$ replicas in the current cluster have acknowledged it. Until then, the version remains unreplicated and may be subject to further propagation or overwrite depending on leader transitions.

If $RF>2$, the replicas are advised to mark their copies replicated once the client has been informed of the success.  Figure~\ref{fig:write-path} illustrates the entire write path from client to leader to replicas and back to client for RF=3.
Note that if $RF=2$, the replica marks its copy replicated right away. We illustrate how each copy marks itself replicated one after the other in Figure~\ref{fig:rf2replprogress}.

\begin{figure}[t]
  \centering
  \includegraphics[width=\linewidth]{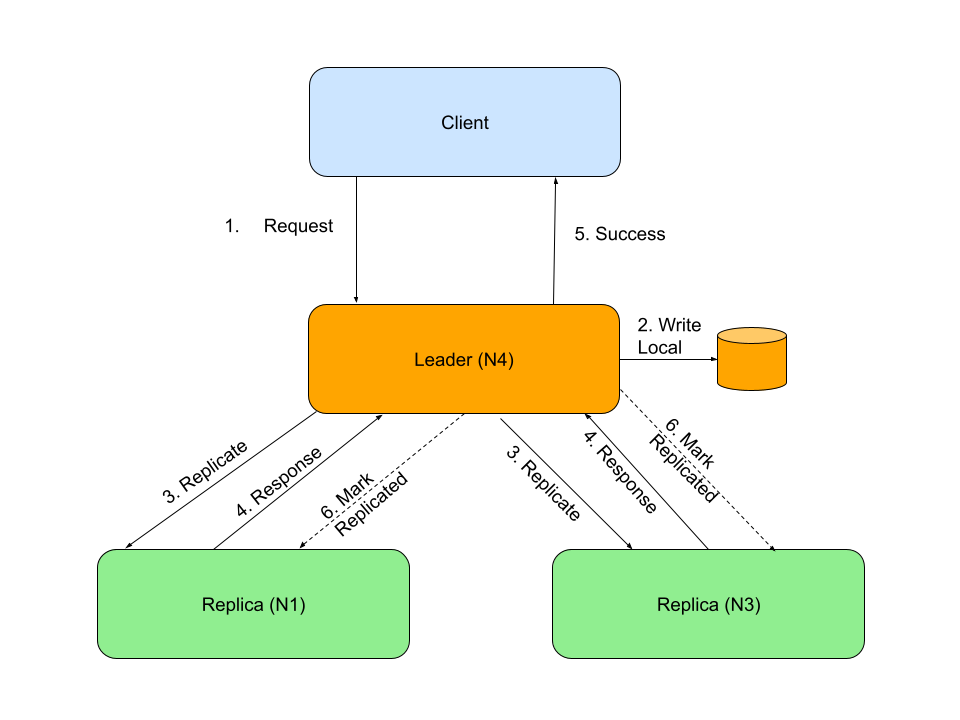}
  \caption{The Write Path across Nodes for a Client Request 
  }
  \label{fig:write-path}
\end{figure}

\begin{figure}[t]
  \centering
  \includegraphics[width=\linewidth]{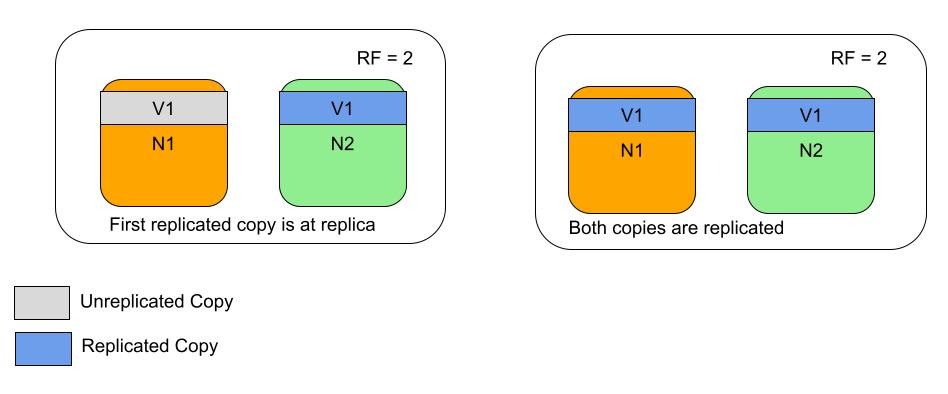}
  \caption{Progression of Replication Across Replicas with time
  }
  \label{fig:rf2replprogress}
\end{figure}

To ensure the leader holds the latest record version, it may first perform \textit{duplicate resolution} (dup-res) to get hold of the most recent version in the cluster. It does this by calling a dup-res function at each node that could hold the latest version of a record. If that version is unreplicated, the leader first re-replicates it to the cluster replicas before applying the update on the latest replicated version of the record.

Reads follow a similar pattern. It may invoke dup-res to ensure it serves a consistent value.

\subsection{Detailed Algorithms for Reads and Writes}
\label{subsec:asyncreadsandwrites}

We now expand the read and write processes into their constituent steps, taking into account that nodes in the system may not have a consistent or synchronized view of the current cluster state. Reads and writes are long-lived operations across multiple nodes in the distributed system and apart from accessing a few shared variables no part of their execution is atomic with other ongoing reclustering or rebalance operations.

Algorithm~\ref{alg:clientwrite} and Algorithm~\ref{alg:replicawrite} describe the write protocol, while Algorithm~\ref{alg:clientread} outlines the read protocol. The \textit{duplicate resolution} (\textit{dup res}) \emph{replica-side handler} is captured in Algorithm~\ref{alg:dupres}; the leader invokes it against candidate holders of the latest version. We now provide additional details about the above algorithms.

\subsubsection{Client-Write Algorithm}
\label{subsubsec:client-write}

Algorithm~\ref{alg:clientwrite} describes the steps taken by the current leader upon receiving a write request from a client. As noted earlier, clients track partition leaders and direct writes accordingly.

In Line 6, the node verifies whether it is still the current leader for the partition. If not, the write is rejected.\footnote{In practice, the request is  proxied to the current leader if known. We omit those implementation details for clarity.}

In Line 9, the leader invokes duplicate resolution (\textit{dup res}) only if it is not full or if the current version of the record does not belong to the current regime. If the leader is full or if it holds a record version from the current regime, \textit{dup res} is unnecessary, as the latest version is guaranteed to be locally available.

If the current version is found to be \textit{unreplicated} (Line 12), the leader triggers a re-replication to all current cluster replicas. This re-replication is treated logically as though the write were issued anew in the current regime—i.e., the re-replicated record is tagged with the current partition regime (PR) as its regime number.

Once the current version has been marked as \textit{replicated}, the leader applies the client’s new write to its local copy.\footnote{We assume for simplicity that the client sends the entire record. In practice, the client may update only a subset of fields.} The new version is then replicated to all cluster replicas.\footnote{To optimize bandwidth, the system may replicate a delta or log describing how to derive the new version from the previous one. If the replica has the prior version, it can immediately apply the delta and discard it. This mechanism is not to be confused with the logs of majority consensus protocols, which rely on ordered, persistent logs and require explicit log management.}

Finally, the leader acknowledges the write to the client only after all replicas have accepted the update.

\subsubsection{Dup Res Algorithm}
\label{subsubsec:dupres}

Duplicate resolution (\textit{dup res}) is executed at the leader node and is conceptually straightforward: the leader queries all nodes that may hold the latest version of a record and selects the version with the largest logical clock (LC), regardless of whether it is marked as \textit{replicated} or \textit{unreplicated}. 

A key aspect of dup res is that a replica responding to such a request must verify that the requesting node is a valid member of the current cluster. Beyond this, no assumptions are made about the relative regimes of the requester and responder. Algorithm~\ref{alg:dupres} captures the dup res process at the replica.

\subsubsection{Replica-Write Algorithm}
\label{subsubsec:replica-write}

Algorithm~\ref{alg:replicawrite} outlines the logic used by cluster replicas when processing write requests received from the leader. 

In high-throughput deployments, partitions may have hundreds or thousands of writes in transit at any given time. During a cluster transition, it is critical to avoid discarding or retrying all such operations. LARK tolerates minor discrepancies between nodes—particularly those differing by at most one regime—to maximize availability without sacrificing linearizability.

The conditions  in the algorithm are evaluated atomically relative to any concurrent changes introduced by reclustering or rebalance and their roles are described first and then we show the necessity for these conditions through some examples in Appendix~\ref{app:replicawrite}.

\begin{itemize}

    \item \textbf{Condition \textsf{LeaderNotTooOld}:} The sender's $PR$ when it received this write request from the client is set to $RR$ in Line 4 of the \textsc{Client-Write} algorithm. This must satisfy $RR \ge ER - 1$ at the replica—i.e., the sender cannot be more than one regime behind any future leader. As we will see in the proof, this ensures there is some continuity in terms of nodes seeing the writes and thus it is safe to accept such a potentially older write. 
    
    \item \textbf{Condition \textsf{SameLeaderRegime}:} This condition relaxes Condition \textsf{LeaderNotTooOld} if the leader has remained unchanged between the time the message was sent and received. In this case, we do not need to worry about the older writes from the same current leader - it is still the leader at the current $PR$ at the replica and that is still within one regime of a future leader by virtue of Condition \textsf{LeaderNotTooNew} being True.
    \ASImplementation{We do not have this condition/optimization in our implementation currently. The main motivation for this is to match what Raft does when the leader has not failed - other failures do not change the term and so updates from the same term are accepted. SameLeaderRegime is like the term in Raft}
    
\item \textbf{Condition \textsf{LeaderInCluster}:} The sender of the write request must be in current cluster (in the replica's view).

\item \textbf{Condition \textsf{LeaderNotTooNew}:} We already saw how this worked in tandem with \textsf{SameLeaderRegime}. In addition, this is also used to prevent two future leaders from treading on each other.

\item \textbf{Condition \textsf{NodeInReplicaSet}:} The receiving node must currently be a legitimate cluster replica, according to its local view of the cluster - note by virtue of the fact that it received the replica write from the leader at regime RR, it would be in the replica set for RR but it needs to be in the replica set at the time it accepts the replica write. 

\end{itemize}

\subsubsection{Client-Read Algorithm}
\label{subsubsec:client-read}

The \textsc{client read} algorithm follows the same initial steps as the \textsc{client-write} algorithm: the leader ensures that it holds the latest version of the record. Once that condition is satisfied, the leader must also verify that it is still considered the current leader all the other replicas before responding to the client.

This additional verification step is necessary to guard against cases where another cluster may have formed in the background and successfully completed a write before the current read request was initiated. Ensuring that the leader is still valid at the time of responding preserves real-time ordering guarantees, a requirement for any protocol that implements linearizability, including Paxos and Raft.

\begin{algorithm}[t]
\caption{Client-Write Algorithm}
\label{alg:clientwrite}
\begin{algorithmic}[1]
\Function{Client-Write}{Key, leader, Record}
    \State $P \gets \text{Partition}(Key)$ \Comment{Identify partition based on key}
    \State \textbf{Read Atomically}:
        \State \hspace{1em} $RR \gets$ PR (Partition Regime)
        \State \hspace{1em} $CurrLeader \gets$ Current leader of $P$
    \If{$leader \neq CurrLeader$}
        \State Reject client write
    \EndIf
    \If{$leader$ is not full \textbf{and} $Key.RR \neq RR$}
        \State Perform \textsc{Dup-Res}
    \EndIf
    \If{$Key$ is unreplicated}
        \State $ClusterReplicas \gets$ \textsc{Replicas}$(\texttt{NodesInCluster}, P)$
        \State Rereplicate $Key$ to $ClusterReplicas$
        \State Mark $Key$ as replicated
    \EndIf
    \State Write $Record$ to local copy
    \State $ClusterReplicas \gets$ \textsc{Replicas}$(\texttt{NodesInCluster}, P)$
    \ForAll{$N \in ClusterReplicas$}
        \State Send \textsc{Replica-Write}$(Key, leader, N, RR, Key.LC, LR)$
    \EndFor
    \If{all replicas accept}
        \State Mark $Key$ as replicated
        \State Acknowledge write success to client
        \State Send \textsc{Mark-Replicated} advice to $ClusterReplicas$
    \Else
        \State Remove local copy of record
        \State Mark $Key$ as unreplicated
        \State Reject client write
    \EndIf
\EndFunction
\end{algorithmic}
\end{algorithm}

\begin{algorithm}[t]
\caption{Dup-Res Replica Handler}
\label{alg:dupres}
\begin{algorithmic}[1]
\Function{Dup-Res}{Key, leader}
    \If{$leader \in \texttt{NodesInCluster}$}
        \State Send back record and logical clock (LC) for $Key$
    \Else
        \State Send back failure
    \EndIf
\EndFunction
\end{algorithmic}
\end{algorithm}

\begin{algorithm}[t]
\caption{Replica-Write Algorithm}
\label{alg:replicawrite}
\begin{algorithmic}[1]
\Function{Replica-Write}{Key, leader, Replica, RR, LC, LRM}
    \State $P \gets \text{Partition}(Key)$
    \State \textbf{Compute atomically:}
    \State \hspace{1em} \textsf{LeaderInCluster} $\gets leader \in \texttt{NodesInCluster}$
    \State \hspace{1em} \textsf{NodeInReplicaSet} $\gets Replica \in \textsc{Replicas}(\texttt{NodesInCluster}, P)$
    \State \hspace{1em} \textsf{LeaderNotTooOld} $\gets (RR + 1 \geq ER)$
    \State \hspace{1em} \textsf{SameLeaderRegime} $\gets (LRM == LR)$
    \State \hspace{1em} \textsf{LeaderNotTooNew} $\gets (PR + 1 \geq ER)$
    \If{(\textsf{LeaderNotTooOld} $\lor$ \textsf{SameLeaderRegime}) $\land$ \textsf{LeaderInCluster} $\land$ \textsf{LeaderNotTooNew} $\land$ \textsf{NodeInReplicaSet}}
        \State $CurrLC \gets$ LC of current version of $Key$ on $Replica$
        \If{$LC > CurrLC$}
            \State Accept write
        \Else
            \State Reject write
        \EndIf
    \Else
        \State Reject write
    \EndIf
\EndFunction
\end{algorithmic}
\end{algorithm}

\begin{algorithm}[t]
\caption{Client-Read Algorithm}
\label{alg:clientread}
\begin{algorithmic}[1]
\Function{Client-Read}{Key, leader}
    \State $P \gets \text{Partition}(Key)$
    \If{$leader$ is current leader of $P$}
        \If{$leader$ is not full \textbf{and} $Key.RR \neq PR$}
            \State Perform \textsc{Dup-Res}
        \EndIf
        \State $ClusterReplicas \gets \textsc{Replicas}(\texttt{NodesInCluster}, P)$
        \If{$Key$ is unreplicated}
            \State Rereplicate $Key$ to $ClusterReplicas$
        \EndIf
        \ForAll{$N \in ClusterReplicas$}
            \If{\texttt{check\_regime}$(N, PR)$ fails}
                \State Reject client read
            \EndIf
        \EndFor
        \State Return record to client
    \Else
        \State Reject client read
    \EndIf
\EndFunction
\end{algorithmic}
\end{algorithm}

\subsection{Changes to Roster}
\label{subsec:rosterchanges}

The algorithms described thus far assume that the succession list for each partition remains fixed, that is, the set of nodes used to assign leadership and replica roles does not change. However, in practice, nodes may be added to or removed from the system, leading to updates in the roster and consequently, new succession lists for all partitions. This process is referred to as \textit{reconfiguration} in the literature.

There is a subtle but important distinction between reconfiguration in majority consensus protocols and in our context. In majority-based protocols (e.g., Paxos, Raft), a reconfiguration is required whenever the identity of the $2f+1$ nodes responsible for a partition changes. In contrast, LARK allows any node in the database to serve as a replica for any partition without requiring a formal reconfiguration. Thus, for LARK, reconfiguration specifically refers to changes in the \textit{roster}, i.e., the set of provisioned nodes in the system.

We implement roster changes by assigning each roster a version number and managing transitions through a two-phase commit protocol across all nodes. Once the coordinator of the roster change confirms that all nodes have prepared to adopt the new roster and its version, it sends a commit message to finalize the change.

Any clusters formed after a node observes the commit message, whether the node is the coordinator or not, will automatically begin using the updated roster and associated version.

\subsection{Proof of Correctness and TLA+ Verification}

We provide a complete proof of correctness in Appendix~\ref{app:proof}.
We have also verified LARK with a TLA+ implementation which is available at https://github.com/sesh-aerospike/lark-tla-spec.

\section{Experiments}
\label{sec:expts}

We present two complementary evaluations using \emph{distinct} discrete-event simulators. 
Section~\ref{sec:macro} quantifies \emph{cluster-scale availability} under independent node failures, reporting unavailability and confidence intervals across large configurations. 
Section~\ref{sec:micro} then examines \emph{per-partition dynamics} during a single-node outage and recovery, focusing on throughput and latency differences between LARK and quorum-log protocols.

\subsection{Cluster-scale availability under independent failures}
\label{sec:macro}

We evaluate the availability of \textbf{LARK} under independent node failures against a majority-quorum baseline (Raft/VR-style per-partition consensus), using a discrete-event simulator. The goal is to validate the structural claims that PAC expands availability beyond replica-set majority by reasoning over the database-wide cluster, and this effect scales with replication factor $RF=f{+}1$.

\subsubsection{Methodology}

We simulate a cluster with $n\!=\!155$ nodes and $P\!=\!4096$ partitions. The choice of $n$ is consistent with the cost boundary in Section~\ref{subsec:globalclustering} for $P\!=\!4096$\footnote{At this $P$, global heartbeats $n(n{-}1)$ and the aggregate message cost of $4096$ per-partition elections are comparable; see Section~\ref{subsec:globalclustering}}. For each replication factor $RF\in\{2,3,4\}$ (i.e., $f\in\{1,2,3\}$ tolerated failures), we sweep independent per-node failure probability $p \in [5\!\times\!10^{-5}, 10^{-2}]$ and repeat each configuration across multiple random seeds.

Replica placement per partition is performed so that all nodes are uniformly loaded (and no partition has two replicas on the same node) in both LARK and the baseline; because failures are modeled as i.i.d.\ across nodes, using AZ- or rack-aware placement would not change these availability results. Time advances in discrete \emph{ticks}\footnote{A tick may correspond to, for example, one second of elapsed time.}. At each tick: (i) every \emph{up} node fails independently with probability $p$; if it fails, it enters a \emph{down} state for a fixed downtime of $t_{\text{down}}=10$ ticks; (ii) all \emph{down} nodes decrement their remaining downtime and recover when it reaches zero; (iii) availability is evaluated for all partitions under LARK and the baseline. We use a target horizon of $\texttt{sim\_ticks}$ per run together with \emph{early stopping} to ensure tight estimates: each run proceeds for $T$ ticks where $50{,}000\le T\le 3{,}000{,}000$, and stops as soon as the 95\% CI half-width for the unavailability estimate $\hat U$ (defined in Eq.~\ref{eq:estimator}) falls below $\max(\varepsilon_{\mathrm{abs}}, \varepsilon_{\mathrm{rel}}\hat U)$ with $\varepsilon_{\mathrm{abs}}=5{\times}10^{-6}$ and $\varepsilon_{\mathrm{rel}}=5\%$. We check this condition every $5{,}000$ ticks and require at least $200$ unavailable events before early stopping can trigger. Unless noted, results aggregate \emph{three} independent seeds per $(f,p)$.

A partition is counted \emph{available in LARK} if the PAC \emph{SimpleMajority} holds—i.e., a database majority is present \emph{and} at least one node with the latest committed copy is reachable; the other PAC regimes are disabled, so reported LARK availability is a lower bound. A partition is \emph{available in the baseline} if a majority of its fixed $2f{+}1$ replica set is reachable.

\subsubsection{Estimator and reporting}
We report the estimated fraction of unavailable partitions under each system 
($\hat U_{\mathrm{LARK}}$, $\hat U_{\mathrm{Maj}}$), their ratio 
$\hat U_{\mathrm{Maj}}/\hat U_{\mathrm{LARK}}$ (``improvement factor''), and variability across seeds.
Formally, let $U_t^{\mathrm{sys}}$ be the number of partitions unavailable at tick $t$ for system $\mathrm{sys}\in\{\mathrm{LARK},\mathrm{Maj}\}$.
With $P$ partitions and $T$ ticks actually run, the estimator is
\begin{equation}
\label{eq:estimator}
\hat U_{\mathrm{sys}} \;=\; \frac{1}{PT}\sum_{t=1}^{T} U_t^{\mathrm{sys}},
\qquad \mathrm{sys}\in\{\mathrm{LARK},\mathrm{Maj}\},
\end{equation}
i.e., the fraction of partition-time (partition-ticks) that is unavailable.
Per-run 95\% confidence intervals use a normal approximation with denominator $PT$; we then summarize across seeds by reporting the mean and standard deviation.
(Plots/tables show $\hat U_{\mathrm{sys}}$; these estimate the long-run quantities $U_{\mathrm{sys}}$.)

\subsubsection{Results}

Figure~\ref{fig:unavail-f1} plots unavailability vs.\ node-failure probability $p$ for $RF\!=\!2$ ($f\!=\!1$). Across the entire range, majority-quorum unavailability is about \textbf{$\approx3\times$} LARK’s; the improvement factor is flat between $2.8$ and $3.0$.
\begin{figure}[h]
  \centering
  \includegraphics[width=\linewidth]{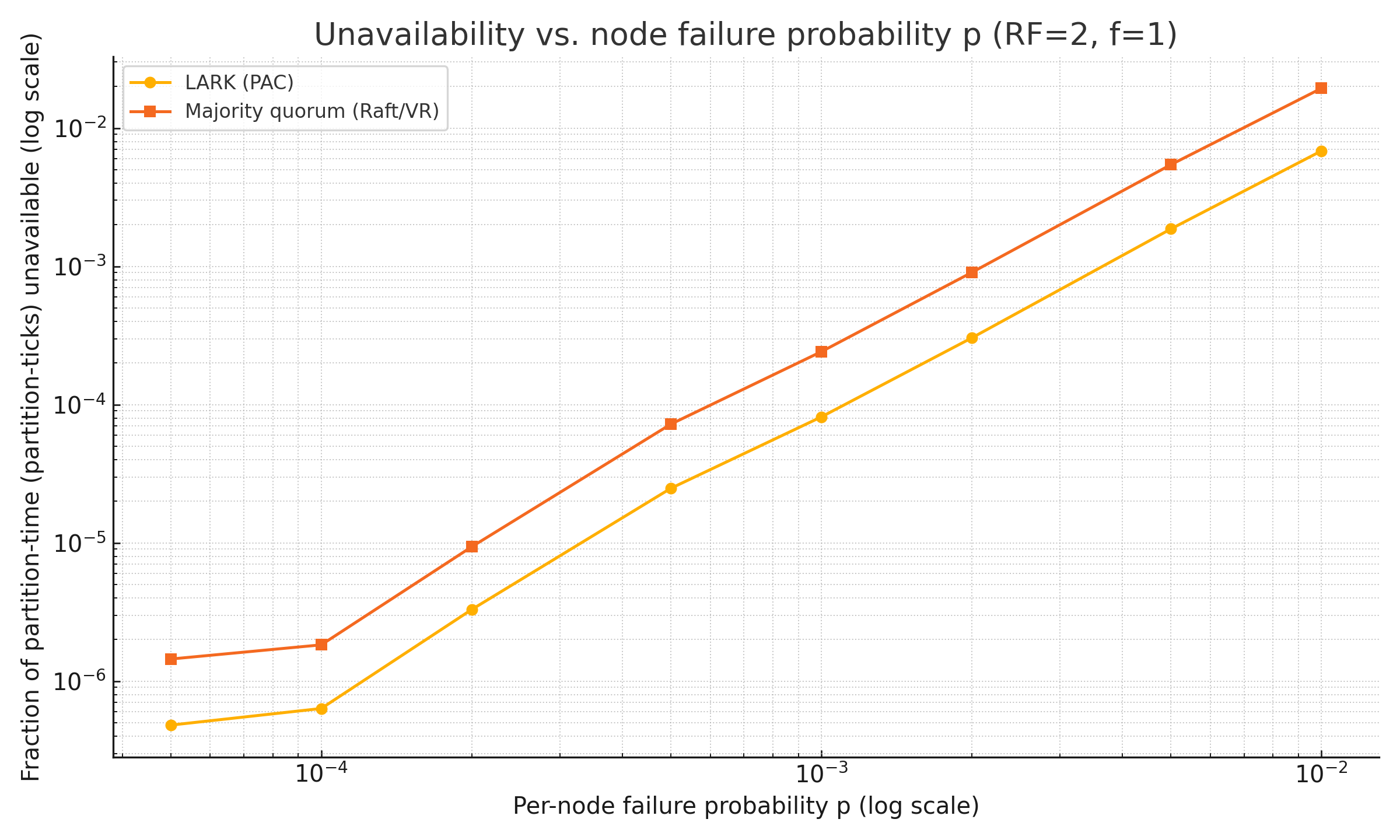}
  \caption{Unavailability vs.\ $p$ for $RF=2$ ($f=1$). Y-axis is the fraction of partition-time (partition-ticks) unavailable.}
  \label{fig:unavail-f1}
\end{figure}

Table~\ref{tab:availability-summary} summarizes the improvement factor (for a few data points) $U_{\mathrm{Maj}}/U_{\mathrm{LARK}}$ across $RF$ values. The simulator reproduces the analytical prediction of Appendix~\ref{app:analytical} that, under independent failures, majority-quorum unavailability scales as $\binom{2f+1}{f+1}p^{\,f+1}$ while LARK scales as $p^{\,f+1}$. The empirical geometric-mean factors observed are \textbf{$\sim$3$\times$} (RF=2), \textbf{$\sim$8–10$\times$} (RF=3), and \textbf{$\sim$36$\times$} (RF=4) as predicted by Equation~\ref{eqn:lark-raft-unavail-ratio} in Appendix~\ref{app:analytical}.

\begin{table}[h]
\centering
\begin{tabular}{c c c c c}
\toprule
$f$ & $RF$ & $p$ & $U_{\mathrm{LARK}}$ & $U_{\mathrm{Maj}}/U_{\mathrm{LARK}}$ \\
\midrule
1 & 2 & $10^{-4}$ & $6.34{\times}10^{-7}$ & $2.89{\times}$ \\
1 & 2 & $10^{-3}$ & $8.15{\times}10^{-5}$ & $2.96{\times}$ \\
1 & 2 & $10^{-2}$ & $6.84{\times}10^{-3}$ & $2.84{\times}$ \\
\midrule
2 & 3 & $2{\times}10^{-4}$ & $5.53{\times}10^{-9}$ & $10.12{\times}$ \\
2 & 3 & $10^{-3}$ & $6.35{\times}10^{-7}$ & $11.62{\times}$ \\
2 & 3 & $10^{-2}$ & $5.70{\times}10^{-4}$ & $8.83{\times}$ \\
\midrule
3 & 4 & $5{\times}10^{-4}$ & $3.26{\times}10^{-10}$ & $70.5{\times}$ \\
3 & 4 & $10^{-3}$ & $5.34{\times}10^{-9}$ & $39.7{\times}$ \\
3 & 4 & $10^{-2}$ & $4.71{\times}10^{-5}$ & $28.7{\times}$ \\
\bottomrule
\end{tabular}
\caption{Selected points from the sweep. Means over seeds; full results in CSV.}
\label{tab:availability-summary}
\end{table}
\begin{table*}[t]
\centering
\small
\setlength{\tabcolsep}{3pt}
\begin{tabular}{r|ccc|ccc|cccc|cccc|cc}
\toprule
\multirow{3}{*}{\#} & 
\multirow{3}{*}{\begin{tabular}{@{}c@{}}$rs$\\(KB)\end{tabular}} & 
\multirow{3}{*}{\begin{tabular}{@{}c@{}}$ps$\\(GB)\end{tabular}} & 
\multirow{3}{*}{\begin{tabular}{@{}c@{}}$bw$\\(MB/s)\end{tabular}} & 
\multicolumn{3}{c|}{Throughput (ops/s)} & 
\multicolumn{4}{c|}{Avg Latency (ms)} & 
\multicolumn{4}{c|}{P99 Latency (ms)} & 
\multicolumn{2}{c}{Recovery} \\
\cmidrule(lr){5-7} \cmidrule(lr){8-11} \cmidrule(lr){12-15} \cmidrule(lr){16-17}
& & & & 
\multirow{2}{*}{LARK} & 
\multirow{2}{*}{BASE} & 
\multirow{2}{*}{Ratio} & 
\multirow{2}{*}{LARK} & 
\multirow{2}{*}{BASE} & 
\multirow{2}{*}{Ratio} & 
\multirow{2}{*}{Delta} & 
\multirow{2}{*}{LARK} & 
\multirow{2}{*}{BASE} & 
\multirow{2}{*}{Ratio} & 
\multirow{2}{*}{Delta} & 
\begin{tabular}{@{}c@{}}LARK\\Backfill\\(s)\end{tabular} & 
\begin{tabular}{@{}c@{}}BASE\\Down\\(s)\end{tabular} \\
& & & & & & & & & & & & & & & & \\
\midrule
1 & 1 & 0.1 & 5 & 2500 & 2364 & 1.06 & 1.1 & 1.0 & 1.07 & +0.1 & 2 & 1 & 2.00 & +1 & 66 & 20 \\
2 & 1 & 0.1 & 48 & 25000 & 24839 & 1.01 & 1.0 & 1.0 & 1.00 & +0.0 & 1 & 1 & 1.00 & +0 & 8 & 2 \\
3 & 1 & 0.9 & 5 & 2500 & 1356 & 1.84 & 1.0 & 1.0 & 1.00 & +0.0 & 1 & 1 & 1.00 & +0 & 135 & 200 \\
4 & 1 & 0.9 & 48 & 25000 & 23640 & 1.06 & 1.0 & 1.0 & 1.00 & +0.0 & 1 & 1 & 1.00 & +0 & 66 & 20 \\
5 & 1 & 9.3 & 5 & 2500 & 837 & 2.99 & 1.0 & 1.0 & 1.00 & +0.0 & 1 & 1 & 1.00 & +0 & 149 & 300 \\
6 & 1 & 9.3 & 48 & 25000 & 13547 & 1.85 & 1.0 & 1.0 & 1.00 & +0.0 & 1 & 1 & 1.00 & +0 & 135 & 200 \\
7 & 10 & 0.1 & 5 & 250 & 236 & 1.06 & 2.1 & 2.0 & 1.07 & +0.1 & 3 & 2 & 1.50 & +1 & 65 & 20 \\
8 & 10 & 0.1 & 48 & 2500 & 2484 & 1.01 & 1.0 & 1.0 & 1.01 & +0.0 & 1 & 1 & 1.00 & +0 & 8 & 2 \\
9 & 10 & 0.9 & 5 & 250 & 136 & 1.84 & 2.2 & 2.0 & 1.10 & +0.2 & 3 & 2 & 1.50 & +1 & 135 & 200 \\
10 & 10 & 0.9 & 48 & 2500 & 2364 & 1.06 & 1.1 & 1.0 & 1.07 & +0.1 & 2 & 1 & 2.00 & +1 & 66 & 20 \\
11 & 10 & 9.3 & 5 & 250 & 84 & 2.98 & 2.2 & 2.0 & 1.11 & +0.2 & 3 & 2 & 1.50 & +1 & 149 & 300 \\
12 & 10 & 9.3 & 48 & 2500 & 1356 & 1.84 & 1.0 & 1.0 & 1.00 & +0.0 & 1 & 1 & 1.00 & +0 & 135 & 200 \\
\bottomrule
\end{tabular}
\caption{Config 2: 80\% read workload with $u=0.5$, $lf=0.5$. Throughput measured until LARK completes backfill; BASELINE extrapolated to same time. Delta shows LARK -- BASELINE (positive means LARK is worse). LARK backfill is the duration from node recovery ($t{=}302$s) to backfill completion. BASELINE downtime is from failure ($t{=}2$s) to migration completion.}
\label{tab:config2_u05_lf05}
\end{table*}
\subsubsection{Discussion and limitations}
These experiments isolate the \emph{structural} availability effects of PAC vs.\ replica-set majority under independent node failures. They do not rely on implementation-specific optimizations. The results corroborate the analysis: LARK’s unavailability scales like $p^{\,f+1}$ with a constant factor advantage of $\binom{2f+1}{f+1}$ over majority-quorum baselines at small $p$, and the advantage persists across the range we swept. 

\subsection{Per-partition throughput and latency during a single failure}
\label{sec:micro}

We complement the availability study with a \emph{per-partition} micro-simulator that drives a single failure to complete recovery timeline and records throughput and latency over time.

\subsubsection{Methodology}
We model a single partition through a failure to complete recovery timeline and record per-second throughput and latency. The replication factor is $RF{=}2$ (tolerates one failure), a common cost/performance point. Time advances in a discrete-event engine with a 1~ms tick; the simulator processes all events at millisecond granularity. We plot aggregated metrics in seconds for readability.

\paragraph{System parameters.}
Each partition has a bandwidth budget $bw\!\in\!\{5,50\}\,\mathrm{MB/s}$; these values arise by dividing a $\sim\!10$~Gb/s ($\sim\!1$~GB/s)\footnote{Throughout, we round KB/MB/GB to powers of 10 for readability.} NIC across $\sim$20--200 partitions per node.\footnote{Given Aerospike's fixed 4096 partitions, $20$ partitions per node corresponds to $\approx 4096/20\approx 200$ nodes, while $200$ partitions per node corresponds to $\approx 4096/200\approx 20$ nodes. This mapping is only used to motivate $bw$; the micro simulator does not model the full cluster.} We fix $\mathrm{RTT}=1\,\mathrm{ms}$; the service time per request is
$\max\!\bigl(1\,\mathrm{ms},\ \text{bytes}/bw\bigr)$,
i.e., the larger of the RTT and the transmission time at the allotted per-partition bandwidth. Requests are scheduled using processor sharing: all in-flight operations share bandwidth equally.

\paragraph{Systems.}
\emph{Baseline (quorum-log, equal storage):} provisions exactly $f{+}1$ data replicas with no spare. When a replica fails at $t{=}2$s, it \emph{immediately} hydrates a replacement on another node via a full-partition transfer and pauses new commits during this rebuild. In our settings, hydration typically completes before the failed node returns at $t{=}302$s; when the original comes back it is no longer a replica and no further data motion is triggered.

\emph{LARK:} does \emph{not} start migration on failure. Because the failed node is the roster replica, LARK continues to commit to the surviving roster replica and a second node (the spare), and waits for the failed node to return. Upon return at $t{=}302$s, LARK backfills only the keys written during the outage to restore the roster placement; this backfill runs in the background while serving continues at full bandwidth. In practice, migrations are delayed briefly to ride out transient flaps; we model this with a 300s delay.

\paragraph{Workload and knobs.}
\begin{itemize}
  \item Record sizes $rs\!\in\!\{1,10\}\,\mathrm{KB}$ (typical for real-time Aerospike workloads).
  \item Partition sizes $ps\!\in\!\{0.1,1,10\}\,\mathrm{GB}$ (with 4096 partitions, $\sim$0.4--41~TB total).
  \item Read/write mix: uniform inter-arrival times with an 80\%/20\% read:write ratio. Reads transfer the full record; writes transfer $\mathsf{lf} \times rs$ bytes, where $\mathsf{lf}$ is the log-bytes fraction.
  \item Offered load $u\!\in\!\{0.5,0.8\}$, chosen to exercise contention during backfill without saturating the links. The arrival rate is computed as $\lambda = u \times bw / \text{avg\_request\_size}$, where avg\_request\_size accounts for the read/write mix and $\mathsf{lf}$.
  \item Log-bytes fraction $\mathsf{lf}\!\in\!\{0.5,1.0\}$ models how much of a record must be transmitted to represent an update (client$\rightarrow$leader and leader$\rightarrow$replicas). $\mathsf{lf}{=}1.0$ represents full replication; $\mathsf{lf}{=}0.5$ represents partial replication (e.g., only logging deltas or metadata).
  \item Failure timeline: node 0 fails at $t{=}2$s and recovers at $t{=}302$s (typical VM/instance reboot duration of 300s).
  \item Simulation duration: 1000s to ensure both LARK and BASELINE complete their recovery/backfill operations.
\end{itemize}

\paragraph{Throughput calculation.}
Both systems run for 1000s to ensure completion of all recovery operations. We report throughput over a measurement window $W$ defined as LARK's backfill completion time (typically 310--501s depending on partition size and bandwidth). For both systems, throughput is computed as total requests completed during $[0, W]$ divided by $W$, directly measured from per-second simulation logs. This ensures both systems are compared over the same time horizon, capturing LARK's availability advantage during the failure period and BASELINE's downtime.

\subsubsection{Results: Low utilization ($u{=}0.5$, $\mathsf{lf}{=}0.5$)}
Table~\ref{tab:config2_u05_lf05} shows results for $u{=}0.5$ and $\mathsf{lf}{=}0.5$. At this utilization, both systems have ample headroom to handle transient load variations.

\paragraph{Throughput.}
LARK achieves 1.01--2.99$\times$ BASELINE's throughput. The advantage is most pronounced when BASELINE's downtime is long (rows 5, 11: 300s downtime $\rightarrow$ 2.99$\times$ and 2.86$\times$ ratios) and minimal when downtime is short (rows 2, 8: 2s downtime $\rightarrow$ 1.01$\times$ ratio). This is expected: LARK maintains availability during the entire failure period ($t{=}2$ to $t{=}302$), while BASELINE is down for 2--300s depending on partition size and bandwidth. LARK's backfill duration (8--149s) is consistently shorter than or comparable to BASELINE's downtime for medium and large partitions.

\paragraph{Latency.}
LARK's average and P99 latencies are nearly identical to BASELINE's (ratios of 1.00--1.11$\times$, deltas of 0.0--0.2ms). This is because the 50\% utilization provides sufficient headroom: even during backfill, LARK allocates 80\% of bandwidth to foreground traffic (4MB/s or 40MB/s), leaving enough capacity to avoid queueing. The small P99 increases (+1ms) are due to occasional transient queues when the probabilistic read/write mix temporarily increases write load.

\subsubsection{Results: High utilization ($u{=}0.8$, $\mathsf{lf}{=}1.0$)}
Table~\ref{tab:config1_u08_lf10} shows results for $u{=}0.8$ and $\mathsf{lf}{=}1.0$. At this higher utilization with full replication, the systems operate closer to capacity.

\begin{table*}[t]
\centering
\small
\setlength{\tabcolsep}{3pt}
\begin{tabular}{r|ccc|ccc|cccc|cccc|cc}
\toprule
\multirow{3}{*}{\#} & 
\multirow{3}{*}{\begin{tabular}{@{}c@{}}$rs$\\(KB)\end{tabular}} & 
\multirow{3}{*}{\begin{tabular}{@{}c@{}}$ps$\\(GB)\end{tabular}} & 
\multirow{3}{*}{\begin{tabular}{@{}c@{}}$bw$\\(MB/s)\end{tabular}} & 
\multicolumn{3}{c|}{Throughput (ops/s)} & 
\multicolumn{4}{c|}{Avg Latency (ms)} & 
\multicolumn{4}{c|}{P99 Latency (ms)} & 
\multicolumn{2}{c}{Recovery} \\
\cmidrule(lr){5-7} \cmidrule(lr){8-11} \cmidrule(lr){12-15} \cmidrule(lr){16-17}
& & & & 
\multirow{2}{*}{LARK} & 
\multirow{2}{*}{BASE} & 
\multirow{2}{*}{Ratio} & 
\multirow{2}{*}{LARK} & 
\multirow{2}{*}{BASE} & 
\multirow{2}{*}{Ratio} & 
\multirow{2}{*}{Delta} & 
\multirow{2}{*}{LARK} & 
\multirow{2}{*}{BASE} & 
\multirow{2}{*}{Ratio} & 
\multirow{2}{*}{Delta} & 
\begin{tabular}{@{}c@{}}LARK\\Backfill\\(s)\end{tabular} & 
\begin{tabular}{@{}c@{}}BASE\\Down\\(s)\end{tabular} \\
& & & & & & & & & & & & & & & & \\
\midrule
1 & 1 & 0.1 & 5 & 3326 & 3153 & 1.05 & 3.4 & 2.5 & 1.38 & +0.9 & 27 & 5 & 5.40 & +22 & 69 & 20 \\
2 & 1 & 0.1 & 48 & 33327 & 33118 & 1.01 & 3.2 & 2.4 & 1.35 & +0.8 & 5 & 4 & 1.25 & +1 & 8 & 2 \\
3 & 1 & 0.9 & 5 & 3316 & 1926 & 1.72 & 4.8 & 2.5 & 1.95 & +2.4 & 28 & 5 & 5.60 & +23 & 172 & 200 \\
4 & 1 & 0.9 & 48 & 33275 & 31535 & 1.06 & 4.0 & 2.3 & 1.74 & +1.7 & 28 & 4 & 7.00 & +24 & 69 & 20 \\
5 & 1 & 9.3 & 5 & 3313 & 1330 & 2.49 & 5.2 & 2.5 & 2.09 & +2.7 & 28 & 5 & 5.60 & +23 & 197 & 300 \\
6 & 1 & 9.3 & 48 & 33187 & 19248 & 1.72 & 5.4 & 2.3 & 2.30 & +3.0 & 29 & 4 & 7.25 & +25 & 171 & 200 \\
7 & 10 & 0.1 & 5 & 332 & 315 & 1.05 & 3.9 & 3.3 & 1.20 & +0.7 & 20 & 11 & 1.82 & +9 & 69 & 20 \\
8 & 10 & 0.1 & 48 & 3333 & 3312 & 1.01 & 2.6 & 2.5 & 1.04 & +0.1 & 6 & 5 & 1.20 & +1 & 8 & 2 \\
9 & 10 & 0.9 & 5 & 331 & 193 & 1.72 & 4.9 & 3.3 & 1.50 & +1.6 & 24 & 11 & 2.18 & +13 & 172 & 200 \\
10 & 10 & 0.9 & 48 & 3326 & 3153 & 1.05 & 3.4 & 2.5 & 1.38 & +0.9 & 27 & 5 & 5.40 & +22 & 69 & 20 \\
11 & 10 & 9.3 & 5 & 331 & 134 & 2.48 & 5.2 & 3.3 & 1.57 & +1.9 & 25 & 11 & 2.27 & +14 & 199 & 300 \\
12 & 10 & 9.3 & 48 & 3316 & 1926 & 1.72 & 4.8 & 2.5 & 1.95 & +2.4 & 28 & 5 & 5.60 & +23 & 172 & 200 \\
\bottomrule
\end{tabular}
\caption{Config 1: 80\% read workload with $u=0.8$, $lf=1.0$. Throughput measured until LARK completes backfill; BASELINE extrapolated to same time. Delta shows LARK -- BASELINE (positive means LARK is worse). LARK backfill is the duration from node recovery ($t{=}302$s) to backfill completion. BASELINE downtime is from failure ($t{=}2$s) to migration completion.}
\label{tab:config1_u08_lf10}
\end{table*}

\paragraph{Throughput.}
LARK achieves 1.01--2.49$\times$ BASELINE's throughput, with the same pattern as the low-utilization case: larger gains when BASELINE's downtime is long (rows 5, 11: 2.49$\times$ and 2.48$\times$) and minimal gains when downtime is short (rows 2, 8: 1.01$\times$). The throughput advantage is purely due to availability: LARK serves requests throughout the failure period while BASELINE is down.

\paragraph{Latency trade-off.}
Unlike the low-utilization case, LARK exhibits significantly higher latencies at high utilization. Average latency increases by 1.04--2.30$\times$ (deltas of +0.1 to +3.0ms), and P99 latency increases by 1.20--7.25$\times$ (deltas of +1 to +25ms). This degradation occurs because LARK operates at 100\% foreground utilization during backfill: with 80\% of bandwidth allocated to foreground (e.g., 4MB/s out of 5MB/s total) and an 80\% offered load, LARK has no headroom to absorb transient load variations. The probabilistic read/write mix creates bursts of writes that temporarily exceed capacity, causing queue buildup and increased latency.

The latency penalty is most severe for small records with high bandwidth (rows 4, 6: 7.00--7.25$\times$ P99 ratio), where the high request rate amplifies queueing effects. Larger records with lower bandwidth (rows 7, 9, 11) show more modest P99 increases (1.82--2.27$\times$) because the lower request rate reduces contention.

\paragraph{Recovery time comparison.}
LARK's backfill duration (8--199s) is consistently shorter than BASELINE's downtime for medium and large partitions (200--300s), but longer for small partitions (2--20s). This reflects the fundamental trade-off: LARK optimizes for availability during failures at the cost of longer recovery for small partitions, while BASELINE optimizes for fast recovery at the cost of downtime.

\subsubsection{Summary}
LARK provides 1.01--2.99$\times$ better throughput than BASELINE by maintaining availability during failures. The throughput advantage scales with BASELINE's downtime: up to 2.99$\times$ when BASELINE is down for 300s (large partitions, low bandwidth) and only 1.01$\times$ when downtime is 2s (small partitions, high bandwidth). At low utilization ($u{=}0.5$), LARK achieves this with negligible latency impact (1.00--1.11$\times$ avg, 1.00--2.00$\times$ P99). At high utilization ($u{=}0.8$), LARK trades latency for availability: average latency increases by 1.04--2.30$\times$ and P99 by 1.20--7.25$\times$, particularly for small-record, high-bandwidth workloads.

\section{Related Work}
\label{sec:related}

The classical lineage for linearizable replication comprises Paxos, Raft, Viewstamped Replication (VR), and Zab (ZooKeeper’s atomic broadcast)~\cite{lamport2001paxos,ongaro2014search,liskov12vr,hunt2010zookeeper}. These systems are \emph{quorum-log} and coordinate via majority quorums over \emph{fixed replica sets}, which can strand partitions even when the cluster at large is healthy. This contrasts with LARK’s combination of a log-free data path and \emph{Partition Availability Conditions} (PAC) that reason at cluster scope.

\paragraph{Quorum-log refinements.}
A number of works make quorum formation or reconfiguration more flexible while remaining quorum-log and replica-set–scoped. \emph{Flexible Paxos} relaxes quorum intersection requirements to reduce quorum sizes but still relies on a persistent log and per-group quorum reasoning~\cite{howard2016flexiblepaxos}. \emph{Vertical Paxos} and \emph{Matchmaker Paxos} decouple or virtualize reconfiguration to simplify membership change and placement of acceptors/learners~\cite{lamport2009verticalpaxos,whittaker2020matchmaker}. These directions are complementary to LARK’s control-plane choices but do not provide PAC-style availability envelopes over the database-wide cluster. Representative Raft optimizations (e.g., KV-Raft, BUC-Raft, RaftOptima) reduce latency or improve log management, but they remain quorum-log designs—retaining ordered logs, majority-quorum intersection, and per-partition leader-coordination rounds—so availability and \emph{immediate partition readiness} after leader changes remain limited.

\paragraph{Log-free, state-direct approaches.}
Closer to LARK’s state-direct path, \emph{CASPaxos} and \emph{linearizable CRDT-based SMR} remove ordered logs and update state directly~\cite{rystsov2018caspaxos,skrzypczak2019linearizablecrdt}. CASPaxos preserves Paxos-style quorums and, under contention, may require multiple round trips or retries for a hot key; availability still depends on a majority of a configured replica set. Linearizable CRDTs typically restrict operation sets or add coordination to ensure convergence and linearizability; again, availability remains tied to the replica set. LARK differs in two respects: (i) it is log-free while keeping a one–round-trip common-case write path, using per-key duplicate resolution only when needed; and (ii) it broadens availability via PAC by reasoning over the cluster as a whole rather than a fixed per-partition replica set.

\paragraph{Fast-path commit and shared logs.}
\emph{CURP} (Consistent Unordered Replication Protocol) decouples client-perceived commit from ordering, finishing many operations quickly with witnesses while pushing to a log in the background~\cite{park2019curp}. CURP remains log-backed and replica-set–majority–based for durability and availability. \emph{Delos} virtualizes consensus atop a shared-log substrate to separate control (reconfiguration) from data (log appends), reducing some catch-up costs while staying fundamentally log-centric and quorum-based~\cite{balakrishnan2020delos}. In contrast, LARK removes logs from the data path and widens availability with PAC.

\paragraph{Production practice versus prototypes.}
Operationally oriented Raft variants (e.g., KV-specialized batching/compaction or alternative reconfiguration procedures) report latency/throughput improvements but stay within the log-and-quorum template. To our knowledge, there are few peer-reviewed reports of \emph{log-free} \emph{and} \emph{cluster-scope} availability (PAC-like) in production. Conversely, widely deployed commercial systems often disclose only partial details, limiting rigorous apples-to-apples evaluation.

\paragraph{Summary.}
Prior work either (a) keeps logs and replica-set majorities (Paxos/Raft/VR/Zab, Flexible/Vertical/Matchmaker Paxos, Delos, CURP), or (b) removes logs but still reasons over replica-set quorums (CASPaxos, linearizable CRDT SMR). LARK’s contribution is orthogonal: \emph{log-free, state-direct replication combined with PAC’s cluster-wide availability reasoning}, with a per-key duplicate resolution step that preserves linearizability without ordered logs.

\section{Conclusion}
\label{sec:conclusion}

We presented \textbf{LARK}, a synchronous replication design for real-time databases that delivers linearizability while minimizing latency and infrastructure cost and, crucially, enlarging the conditions under which partitions remain available. LARK combines three elements:
(i) \emph{Partition Availability Conditions (PAC)}, which reason over the database-wide cluster rather than a fixed replica set;
(ii) a \emph{log-free} read/write path with per-key duplicate resolution and background migration, making leaders immediately ready across transitions instead of waiting for ordered-log catch-up; and
(iii) tolerance of \emph{bounded view skew} (at most one regime), which keeps writes flowing during leader changes and trims tail latencies.

We established safety via formal arguments and a TLA+ specification, and we quantified benefits with analysis and simulation. Under independent failures, LARK’s unavailability scales as $p^{f+1}$ with a constant-factor advantage (e.g., $\sim\!3\times$ at $RF{=}2$, $\sim\!8$–$10\times$ at $RF{=}3$) over majority-quorum baselines. Under equal storage budgets, LARK continues committing during data-node failures while quorum-log systems pause to hydrate a replacement voter. Per-partition micro-experiments show that LARK maintains throughput during single-node outages, matching baseline latencies at moderate load and trading some latency for uninterrupted availability at high load.

There are two areas of future work we have identified:
\begin{enumerate}[leftmargin=1.2em]
  \item \textbf{Roster reconfiguration.} Streamline the roster-change path (Section~\ref{subsec:rosterchanges}) to reduce activation latency while preserving PAC semantics and safety.
  \item \textbf{Scaling clusters.} Replace full-mesh heartbeats with localized membership for groups of partitions (“partition clusters”), retaining PAC’s cluster-wide reasoning while lowering global reclustering pressure.
\end{enumerate}

\begin{acks}
We used AI-assisted tools for writing and engineering support. Specifically, ChatGPT for wording/grammar edits and figure/table captions; and Cursor—using GPT-5 and Claude Sonnet 4.5 models—for simulator coding assistance (e.g., boilerplate, refactoring, and debugging suggestions). 
\end{acks}


\bibliographystyle{ACM-Reference-Format}
\bibliography{myrefs}

\appendix
\section{Replica Write Algorithm}
\label{app:replicawrite}

We now present a few illustrative examples that highlight the necessity of enforcing Conditions \textsf{LeaderInCluster}, \textsf{LeaderNotTooOld}, \textsf{LeaderNotTooNew} and \textsf{NodeInReplicaSet} in the \textsc{replica-write} algorithm. In each example, we name nodes as $N1$, $N2$, $N3$, etc., and indicate whether a node is \textit{full} in parentheses. We focus on a single partition and assume that the succession list for the entire roster follows lexicographic order. Time flows top to bottom. 

\par\noindent\rule{0.5\textwidth}{0.2pt}
\textbf{Example 1: Necessity of Condition \textsf{LeaderInCluster}} \\
\begin{verbatim}
RF = 2, Nodes: N1, N2, N3

Cluster = {N1 (full), N3 } 
// PR=ER=1 at N1 and N3
N1 receives a client write for version V
N1 writes to local copy with RR=1
Replica write for V to N3 is delayed

Cluster = {N2 , N3} // N1 is not in cluster
// PR=ER=2 at N2 and N3
N2 becomes leader and receives a write for V'
N2 performs dup res with N3 
Delayed write for V arrives at N3
\end{verbatim}

Conditions  \textsf{LeaderNotTooOld, LeaderNotTooNew} and \textsf{NodeInReplicaSet}  are satisfied at $N3$ when the delayed replica write for version $V$ arrives (last line in the example above). If Condition \textsf{LeaderInCluster} were not enforced, this write would be accepted. However, $N2$, as the leader, would be unaware of version $V$ and, having just completed a dup res, could proceed to process a client write under the incorrect assumption that it held the latest version.

\par\noindent\rule{0.5\textwidth}{0.2pt}
\textbf{Example 2: Necessity of Condition \textsf{LeaderNotTooOld}.} \\
\begin{verbatim}
RF = 3, Nodes: N1, N2, N3, N4, N5

Cluster = {N1 (full), N3, N4, N5} // N2 down
// PR = 1 for N1, N3 and N4
N1 receives a client write for version V
N1 writes to local copy with RR=1
Replica write for V to N4 is acked
Replica write for V to N3 is delayed

Cluster = {N1 (full), N2, N3} // N4, N5 down
// PR = 2 for N1, N2, N3

Cluster = {N2, N3, N5} // N1, N4 down
// PR = 3 at N2 and N5
// PR = 2 and ER = 3 at N3 (not yet rebalanced)
N2 becomes leader and receives a write for V'
Dup res succeeds at N3 (N2 was in N3's cluster in PR = 2)
Dup res succeeds at N5 (N2 in N5's cluster in PR = 3)
Replica write for V arrives at N3 and is accepted
\end{verbatim}

Conditions  \textsf{LeaderInCluster}, \textsf{LeaderNotTooNew}  and \textsf{NodeInReplicaSet} are satisfied at $N3$ when the delayed replica write for version $V$ arrives (last line in the example above). If Condition \textsf{LeaderNotTooOld} were not enforced, this write would be accepted. However, $N2$, as the leader, would be unaware of version $V$ and, having just completed a dup res, could proceed to process a client write under the incorrect assumption that it held the latest version.

\par\noindent\rule{0.5\textwidth}{0.2pt}

\textbf{Example 3: Necessity of Condition \textsf{LeaderNotTooNew}.} \\
\begin{verbatim}
RF = 3, Nodes: N1, N2, N3, N4, N5

Cluster = {N1 (full), N2, N3, N4, N5}
//PR = 1 at N1, N2 and N3

Cluster = {N2, N3, N4} // N1, N5 down
// PR = 2 at N2, N4
// PR = 1; ER = 2 at N3 (not yet rebalanced)
N2 receives a client write for version V
N2 issues dup res to N3 and N4 - succeeds
N2 issues write for V - N4 acks
Write of V to N3 is delayed


Cluster = {N1, N3, N5} // N2, N4 not in cluster
// PR = 3 at N1, N5
// PR = 1 and ER = 3 at N3 (still not rebalanced)
N1 becomes leader and receives client write for V'
Dup res succeeds at N3 (N1 in cluster when PR = 1)
Dup res succeeds at N5 (N1 in cluster when PR = 3)
Replica write for V arrives at N3 and is accepted
\end{verbatim}

Conditions  \textsf{LeaderInCluster}, \textsf{LeaderNotTooOld}  and \textsf{NodeInReplicaSet} are satisfied at $N3$ when the delayed replica write for version $V$ arrives (last line in the example above). If Condition \textsf{LeaderNotTooNew} were not enforced, this write would be accepted. However, $N1$, as the leader, would be unaware of version $V$ and, having just completed a dup res, could proceed to process a client write under the incorrect assumption that it held the latest version.

\par\noindent\rule{0.5\textwidth}{0.2pt}
\textbf{Example 4: Necessity of Condition \textsf{NodeInReplicaSet}} \\
\begin{verbatim}
RF = 2, Nodes: N1, N2, N3, N4

Cluster = {N1, N2, N3, N4} 
// PR=ER=1 at N1, N2

Cluster = {N1, N4} 
// PR=ER=2 at N1
// PR=1 at N4, ER = 2 at N4
N1 receives a client write for V
N4 accepts replica write for V (Problem!)

Cluster = {N2, N3, N4}
//PR=ER=3 at N2, N3, N4
N4 never rebalanced to PR=2 
So N4 does not think it is a duplicate
N2 will not dup res with N4
N2 can write V' without seeing V

\end{verbatim}

At the time N4 acceptes replica write for V, Conditions \textsf{LeaderNotTooOld}, \textsf{LeaderInCluster} and \textsf{LeaderNotTooNew} are all satisifed at N4 when its PR=1. However, it was not a replica then and as a result not a duplcate and therefore N2 does not dup res wit N4 when PR=3 causing N2 to miss V.

\par\noindent\rule{0.5\textwidth}{0.2pt}

\section{Formal Proof of Correctness}
\label{app:proof}

We first prove the Lemmas of Section~\ref{sec:partitionavailability}.

\begin{lemma}
\label{lemma:onerosterreplica}
Any cluster that satisfies one of the PAC rules for a given partition must include at least one roster replica of that partition.
\end{lemma}
\begin{proof}
This is directly enforced by the PAC rules:
\begin{itemize}
    \item \textit{AllRosterReplicas}, \textit{SimpleMajority} and \textit{HalfRoster} require roster replica inclusion by definition.
    \item \textit{SuperMajority} implies fewer than $RF$ nodes are missing, so at least one roster replica is present.
\end{itemize}
\end{proof}

\begin{lemma}
\label{lemma:onecommonnode}
Let $C_1$ and $C_2$ be two distinct clusters that both satisfy PAC for a partition. Then $C_1$ and $C_2$ must share at least one node.
\end{lemma}
\begin{proof}
We analyze this based on the condition satisfied by $C_1$:
\begin{itemize}
    \item If $C_1$ satisfies SuperMajority, then it must intersect with any other majority-based cluster ($C_2$ satisfying SuperMajority, SimpleMajority, or HalfRoster). If $C_2$ satisfies AllRosterReplicas, then by Lemma~\ref{lemma:onerosterreplica}, they share a roster replica.
    \item If $C_1$ satisfies AllRosterReplicas, then any $C_2$ satisfying PAC will contain a common roster replica by Lemma~\ref{lemma:onerosterreplica}.
     \item If $C_1$ satisfies SimpleMajority: Similar argument as SuperMajority case.
     \item If $C_1$ satisfies HalfRoster: If $C_2$ is SuperMajority or SimpleMajority then they will share a node in common. If $C_2$ is AllRosterReplicas by Lemma~\ref{lemma:onerosterreplica}, they share a node in common.  Finally, if $C_2$ is HalfRoster they share the cluster leader. 
\end{itemize}
\end{proof}

\begin{lemma}
\label{lemma:onlyonecluster}
During any regime, there is at most one cluster in the system that satisfies PAC for a given partition.
\end{lemma}
\begin{proof}
Assume two clusters $C_1$ and $C_2$ satisfy PAC simultaneously in the same regime. Since cluster membership is determined via a global consensus protocol, the two clusters must be disjoint. But this contradicts Lemma~\ref{lemma:onecommonnode}, which states they must share a node.
\end{proof}

\mycomment{
\begin{lemma}
\label{lemma:monotonicregimenumbers}
The partition regime number (PR) increases monotonically across PAC-valid clusters for a given partition.
\end{lemma}
\begin{proof}
Let $C_1$ and $C_2$ be two PAC-valid clusters formed sequentially, with $C_1$ formed before $C_2$. By Lemma~\ref{lemma:onecommonnode}, there exists a node $N$ common to both clusters. Recall that the reclusterng step determines the maximum of the exchange numbers of the nodes in the cluster, increments this maximum by one and assigns each participating node (in particular, $N$) this new \textit{exchange number}. PR  at $N$ is set to the exchange number at the time of rebalance after reclustering. Hence, the PR in $C_2$ is strictly greater than in $C_1$.
\end{proof}
}

\begin{lemma}
\label{lemma:clusterreplicacommon}
Let $C_1$ and $C_2$ be two clusters available for partition $P$, with regime numbers $R_1$ and $R_2$ such that $R_1 < R_2$ and no intermediate regime exists where $P$ was available. Then at least one of the cluster replicas from $C_1$ is also present in $C_2$.
\end{lemma}
\begin{proof}
If $C_2$ satisfies SimpleMajority or HalfRoster, then it must include a full node from $R_1$, which was a cluster replica in $C_1$. If $C_2$ satisfies AllRosterReplicas, then by Lemma~\ref{lemma:onerosterreplica}, one of the roster replicas from $C_1$ is present in $C_2$. If $C_2$ satisfies SuperMajority, then one of the cluster replicas of $C_1$ will be in $C_2$.
\end{proof}

We now get into proving the reads and writes of LARK algorithm. 

\begin{Definition}
\label{def:replicated}
A version $V$ of a record is said to be \textbf{replicated} if any node in the cluster has marked it as replicated.
\end{Definition}

\begin{lemma}
\label{lemma:RFdups}
At all times, there are at least $RF$ nodes in the roster—at least one of which is a roster replica—that have the latest copy (the copy itself could be replicated or unreplicated) of any \textit{replicated} record and are considered duplicates.
\end{lemma}

\begin{proof}
Once a record version is replicated, by definition, it must have been written to $RF$ nodes, which at that point are all cluster replicas and therefore duplicates. Lemma~\ref{lemma:onerosterreplica} guarantees that at least one of these is a roster replica.

Over time, as nodes are reclustered and rebalance occurs, these nodes may cease being cluster replicas and initiate migrations. However, as described in Section~\ref{subsec:duplicates}, when a node exits the duplicate set via migration, the record version is transferred to the new cluster replicas—maintaining the invariant that $RF$ duplicates exist, including one roster replica.
\end{proof}

\textbf{Proof Roadmap.} The goal of this section is to prove that LARK guarantees linearizability—even under asynchronous execution, partial failures, leader transitions, and message delays. We build the proof on a set of structural invariants that govern the evolution of record versions and the propagation of updates. The overall strategy is as follows:

\begin{itemize}
    \item \textbf{Lemmas~\ref{lemma:dupresR+m}--\ref{lemma:fullseeseverything}} show that once a version is replicated, it will be seen by any future leader before it performs a write. This guarantee is achieved through a combination of duplicate resolution and proactive migration. The proofs are in the Appendix.

    \item \textbf{Theorem~\ref{theorem:lineage}} proves the core safety property: no record version can have two children that are both replicated. This ensures that the version lineage remains a single chain.

    \item \textbf{Theorems~\ref{theorem:writeextendschain} and~\ref{theorem:linearizable}} establish that writes always extend the latest visible version, and reads return values consistent with this version chain—thereby ensuring linearizability.
\end{itemize}

\begin{lemma}
\label{lemma:dupresR+m}
Let $KV$ be a record belonging to partition $P$. Let $L$ be a leader that writes a version $V$ of $KV$ with record regime $R$, and assume that $V$ becomes replicated eventually. Consider a cluster $CL$ with regime $R + m$ (for some $m \geq 1$) satisfying the following conditions:

\begin{itemize}
    \item No writes have occurred to $KV$ since version $V$.
    \item Partition $P$ is available in regime $R + m$.
    \item No node is full for $P$ at the start of regime $R + m$.
    \item $L$ is not the leader of $CL$ for partition $P$.
\end{itemize}

Then, by the end of regime $R+1$ (for $m =1$) or by the beginning of regime $R + m$, (for $m > 1$), there exists at least one \textit{duplicate node} that has seen version $V$.

\end{lemma}
\begin{proof}
Since no node is full at the beginning of regime $R + m$, $CL$ must satisfy one of the following PAC conditions: SuperMajority or AllRosterReplicas.

We consider two major cases:

\begin{casesp}
\item \textbf{There exists at least one node in $CL$ in which $V$ was successfully replicated by the start of regime $R+m$.}

One of the RF nodes mentioned in Lemma~\ref{lemma:RFdups}, will be in CL and will be a duplicate that has seen V at the start of $R+m$.

\item \textbf{No node in $CL$ contains a replicated version of $V$ at the start of $R+m$.}
\label{case:notreplicated}

Let $X_i$ be one of the cluster replicas that accepts $V$ (in line 11 of \textsc{Replica-Write Algorithm}) and is in $CL$ (Such a node will exist as CL is a SuperMajority or AllRosterReplicas). There are two subcases:

\begin{casesp}
    \item \textbf{$X_i$ has not yet accepted $V$ when its ER becomes $R + m$ (as part of reclustering for $CL$).} \\
    In this case, if $m > 1$, $X_i$ will not accept a write for version $V$ with record regime $R$, since Condition C1 (which requires $RR + 1 \geq ER$) of the \textsc{Replica-Write} algorithm will not be satisfied, and neither will Condition C2. Hence, $X_i$ cannot contribute to $V$ becoming replicated, contradicting the assumption that $V$ does eventually get replicated. It also follows from the above that for $v$ to become replicated eventually it has to be accepted by the end of regime $R+1$. 

    \item \textbf{$X_i$ accepted $V$ before its $ER$ became $R + m$.} \\
    In this case, $X_i$ holds an \textit{unreplicated} copy of $V$ at the start of regime $R + m$. $X_i$ has seen V. Either it is a duplicate or  by an argument analogous to Lemma~\ref{lemma:RFdups}, there will be at least $RF$ nodes in the system that have seen the unreplicated version of $V$ and are duplcates, at least one of which is a roster replica - one of these nodes will be in CL.

    As an aside, this unreplicated version may eventually be re-replicated in regime $R + m$ or later. This operation is a no-op from a logical perspective, as the content of the version remains the same; the only difference is that it becomes associated with a new regime. This does not affect the correctness of the protocol.
\end{casesp}
\end{casesp}

In all cases, at least one duplicate in $CL$ has seen version $V$ when regime $R + m$ begins.
\end{proof}

\begin{lemma}
\label{lemma:fullnodeseesR-1}
Let $KV$ be a record belonging to partition $P$ with a replicated version $V$ written by a leader $L$ with record regime $R$. Assume there has been no write to $KV$ since $V$, and that $L$ is not the leader of $P$ in regime $R+1$. Then any node $N$ that becomes full for $P$ at any point during regime $R+1$ is guaranteed to have seen $V$ once $V$ is replicated and $N$ becomes full. Further, a full node $N$ is guaranteed to see $V$ by the end of regime $R+1$.
\end{lemma}
\begin{proof}

We consider two possibilities for node $N$ which is a cluster replica in regime $R+1$:

\begin{casesp}
    \item \textbf{$N$ is a cluster replica in regime $R$.} \\
    In this case, $N$ either receives the replica write for version $V$ during regime $R$, or during regime $R+1$. It cannot be later than regime $R+1$ as Condition C1 of the \textsc{Replica-Write} algorithm has to be satisfied (Condition C2 is not satisfied by assumption as the leader has changed).

    \item \textbf{$N$ is not a cluster replica in regime $R$.} \\
    In this case, $N$ was not full in regime $R$ and becomes full only through migration during regime $R+1$. Let the cluster leader in regime $R+1$ be $M$. By Lemma~\ref{lemma:clusterreplicacommon}, there exists at least one node $X$ that is both a cluster replica in regime $R$ and a member of the cluster in regime $R+1$.  Note that $X$ could be $M$ but that only makes some part of the arguments below no-ops. Note that by \textbf{PR Match for Migration} requirement, $X$ must first update its partition regime to $R+1$ before migrating into $M$, and subsequently into $N$.

    We now consider two subcases, based on when $X$ receives the replica write for $V$:

    \begin{casesp}
        \item \textbf{$X$ receives the replica write for $V$ while its $PR = R$.} \\
        In this case, $X$ sees $V$ before its regime transitions to $R+1$. It will carry $V$ into $M$ during migration, and $M$ will propagate $V$ to $N$. Thus, $N$ sees $V$ upon becoming full.

        \item \textbf{$X$ receives the replica write for $V$ while its $PR = R+1$.} \\
        Let $M'$ be the leader in regime $R$ who wrote $V$. For $X$ to accept a write from $M'$ in regime $R+1$, Condition A of the \textsc{Replica-Write} algorithm requires that $M'$ be part of $X$'s current cluster, making $M'$ a duplicate at the beginning of regime $R+1$.

        Now consider two further subcases, depending on when $M'$ writes its local copy of $V$:
        \begin{itemize}
            \item If $M'$ writes its local copy (line 17 of Algorithm~\ref{alg:clientwrite}) \textit{before} migrating into $M$, then $M$ sees $V$ through migration and propagates it to $N$.
            \item If $M'$ writes its local copy \textit{after} migrating into $M$, then the write must use a replica set corresponding to regime $R+1$ or greater. 
            \begin{itemize}
                \item If the replica set corresponds to regime $R+1$, then $N$ is part of that replica set and will receive the write directly while it is regime $R+1$ (otherwise $N$ will reject the write by condition C1 of \textsc{Replica-Write} algorithm).
                \item If the replica set corresponds to a regime strictly greater than $R+1$, then the replica write will be rejected by Condition C1 of \textsc{Replica-Write} algorithm (as $RR$ is $R$ and $ER$ is greater than $R+1$ at $X$). Note that Condition C2 does not hold by the assumptions of the lemma. This contradicts the assumption that $V$ is successfully written with regime $R$.
            \end{itemize}
        \end{itemize}
    \end{casesp}
\end{casesp}

In all cases, $N$ is guaranteed to have seen $V$ once it becomes full in regime $R+1$.
\end{proof}

\begin{lemma}
\label{lemma:fullseeseverything}
Let $KV$ be a record with a replicated version $V$ with a record regime of $R$. Assume no write to $KV$ has occurred since $V$, and that the leader $L$ of regime $R$ is not the leader of partition $P$ in regime $R+k$ for some $k \geq 2$. Then any node $N$ that becomes full for $P$ at any point during regime $R+k$ is guaranteed to see $V$ as soon as $N$ becomes full.
\end{lemma}
\begin{proof}
We prove the statement by induction on $k$.

\textbf{Base case ($k = 2$):} If any node $N'$ ($N'$ could be $N$) that is part of the cluster in regime $R+2$ was also full in regime $R+1$, then by Lemma~\ref{lemma:fullnodeseesR-1}, $N'$ will have seen $V$ by the end of regime $R+1$. $N'$ will either become the leader or migrate its data into the leader which in turn will migrate into $N$ and therefore $N$ will see $V$. If no node of the cluster is full at the start of regime $R+2$ then all conditions of Lemma~\ref{lemma:dupresR+m} are satisfied and there exists some duplicate node that has seen $V$ at the beginning of regime $R+2$. The cluster leader will perform dup-res with this duplicate node and see $V$ and migrate that into $N$.

\textbf{Inductive step:} Assume the statement holds for some fixed $k$; that is, any node that becomes full during regime $R+k$ will have seen $V$ once $V$ is replicated. We now show that the statement also holds for $k+1$, i.e., for regime $R + (k+1)$.

Let $N$ be a node that becomes full for $P$ during regime $R + (k+1)$. We consider two main cases:

\begin{casesp}
    \item \textbf{$N$ was already full at the end of regime $R + k$.} \\
    By the induction hypothesis, $N$ must have seen $V$.

    \item \textbf{$N$ was not full at the end of regime $R + k$ but becomes full in regime $R + (k+1)$.} \\
    We consider two subcases:
    \begin{casesp}
        \item \textbf{Some node $N'$ was full at the beginning of regime $R + (k + 1)$.} \\
        By the induction hypothesis, $N'$ has seen $V$, as $N$ was full at the end of regime $R+k$. During regime $R + (k+1)$, $N'$ either becomes the leader or migrates its data into the leader. The leader, in turn, either migrates into $N$ or is $N$ itself. Therefore, $N$ will receive the version $V$ through $N'$.

        \item \textbf{No node was full at the beginning of regime $R + (k+1)$.} \\
       Since partition $P$ is available (as $N$ becomes full), all preconditions of Lemma~\ref{lemma:dupresR+m} are satisfied. Thus, the cluster formed in regime $R + (k+1)$ contains at least one \textit{duplicate} node that has seen $V$.

        The leader of this cluster (possibly $N$ itself) will invoke \textsc{Dup-Res} for $KV$ before becoming full. Consequently, it will see $V$, and since $N$ becomes full in this regime (either as leader or via migration from leader), it will also see $V$.
    \end{casesp}
\end{casesp}

In all cases, node $N$ sees $V$ once it becomes full in regime $R + (k+1)$. This completes the inductive step.
\end{proof}

\begin{theorem}
\label{theorem:lineage}
For a system operating under the rules of Section~\ref{subsec:asyncreadsandwrites}, at no point in time can there exist a record $KV$ with a version $V$ that has two distinct children, both of which are replicated.
\end{theorem}

\begin{proof}
\begin{Assumption}

\label{assumption:tworeplicatedchildren}
Assume, for contradiction, that a record $KV$ has a version $V$ with two children $C_1$ and $C_2$, both of which are replicated. Let the record regimes of $C_1$ and $C_2$ be $R_1$ and $R_2$, respectively, with $R_1 < R_2$. Let $C_2$ be the version with the smallest logical clock (LC) among all versions with regime $R_2$.

Assume $RF = k$. Let $X_1, X_2, \ldots, X_k$ be the replicas that participated in $C_1$, with $X_1$ as the leader when $C_1$ was written. Similarly, let $Y_1, Y_2, \ldots, Y_k$ be the replicas that participated in $C_2$, with $Y_1$ as the leader for $C_2$.

Assume $Y_1 \ne X_1$ (i.e., the two leaders are different). Otherwise, $Y_1$ would have seen $C_1$ before writing $C_2$. We will not formally prove that concurrent writes to the same leader will be properly sequenced with regard to their regimes, any reasonable implementation would take care of that.

\end{Assumption}
This scenario is illustrated in Figure~\ref{fig:tworeplicatedchildren}.

\begin{figure}[t]
\centering
\includegraphics[width=1.0\linewidth, height=2.5in]{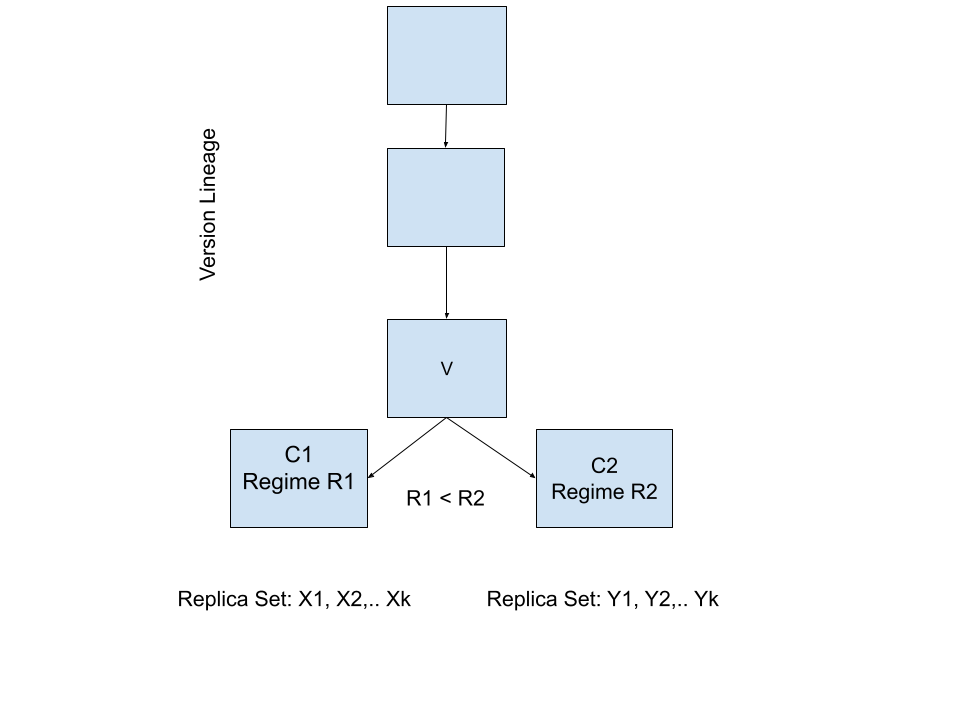}
\caption{Version $V$ with Two Replicated Children $C_1$ and $C_2$}
\label{fig:tworeplicatedchildren}
\end{figure}

We now consider two main cases for how $Y_1$ writes $C_2$.

\begin{casesp}

\item \textbf{$Y_1$ performs a \textsc{Dup-Res} for $KV$ (Line 10, Algorithm~\ref{alg:clientwrite}).} \\

If $R_2 > R_1 + 1$, Lemma~\ref{lemma:dupresR+m} guarantees that the dup-res will find $C_1$, contradicting the assumption. So we focus on the case where $R_2 = R_1 + 1$.

Since $Y_1$ performs dup-res, the cluster ($CL_2$) must satisfy either \textbf{SuperMajority} or \textbf{AllRosterReplicas}. Therefore, at least one cluster replica that participated in $C_1$ (say, some $Z \in \{X_1, \ldots, X_k\}$) must also be in $CL_2$. 

\begin{casesp}
    \item \textbf{$X_1 \notin CL_2$}
    \label{item:x1notincl2}

    $Z$ has to be a duplicate when its $ER=R_2$. If $Z$ ever rebalanced as part of regime $R_1$ it would become a clsuter replica (as per its own view) and therefore a duplicate. If $Z$ was not a cluster replica in $R_1-1$, it could not have accepted the replica write as part of regime $R_1-1$ (Condition \textsf{NodeInReplicaSet} is not satisifed). $Z$ can not accept in a regime lower than $R1-1$ as Condition \textsf{LeaderNotTooNew} will not be satisfied. If $Z$ accepts $C_1$ when its regime is $R_2$, it has to be a replica (as per Condition \textsf{NodeInReplicaSet}) and therefore a duplicate. As a result $Z$ will receive the dup-res from $Y_1$.

    For $Y_1$ to not see $C_1$ during dup-res, the dup-res must occur before $Z$ receives a replica write for $C_1$. Since $Z$ participates in $CL_2$ during dup-res, its exchange regime ($ER$) must be $R_2$ (it cannot be greater than $R_2$ due to Condition C1 of the \textsc{Replica-Write} Algorithm for $C_1$). For the same reasons, the $PR$ at $Z$ can not be greater than $R_2$ when the dup-res from $Y_1$ arrives. We consider three cases based on the value of $PR$:

    \begin{casesp}
        \item \textbf{$PR < R_1$ at $Z$ when dup-res from $Y_1$ arrives}
            The replica write for $C_1$ can not happen when $PR < R_1$ by Condition D of the \textsc{Replica-Write} Algorithm as $ER$ is already $R_2$ - it has to be at a later $PR$. The next rebalance however will make the $PR$ at least $R_2$. It can not be greater than $R_2$ as Condition C1 of the \textsc{Replica-Write} Algorithm will fail for the replica write of $C_1$. Therefore, $PR=R_2$ at $Z$ when the replica write of $C_1$ happens, but this means $X_1$ is in the cluster in regime $R_2$ (by Condition A of the \textsc{Replica-Write} Algorithm). This is a contradiction to the assumption of Case~\ref{item:x1notincl2}.
        
        \item \textbf{$PR = R_1$ at $Z$ when dup-res from $Y_1$ arrives}

        Since dup-res from $Y_1$ to $Z$ succeeds, $Y_1$ must be in $Z$'s cluster in $R_1$. $Y_1$ is a cluster leader in regime $R_2$ and is therefore a roster replica. This implies $Y_1$ will be a cluster replica in regime $R_1$ and will receive the replica write for $C_1$. Contradiction to Assumption~\ref{assumption:tworeplicatedchildren}.

        \item \textbf{$PR = R_2$ at $Z$ when dup-res from $Y_1$ arrives}

        The replica write from $X_1$ must occur in $R_2$, implying $X_1$ is in $Z$'s cluster in $R_2$. Thus, $X_1 \in CL_2$, contradicting the premise that $Y_1 \ne X_1$ and $Y_1$ did not see $C_1$.
    \end{casesp}

    \item \textbf{$Z = X_1$} \\
    
    This implies $Y_1$ issued dup-res to $X_1$ before $X_1$ wrote $C_1$ locally (otherwise $Y_1$ would have seen $C_1$). Let $R_d$ be the partition regime at $X_1$ when dup-res occurs, and $R_w$ the partition regime at $X_1$ when $X_1$ writes $C_1$. Note that $R_d$ can not be less than $R_1$ as it leads to one of two possibilities: 1) $R_w = R_d$ which is less than $R_1$ - this violates Condition \textsf{LeaderNotTooNew} of the \textsc{Replica-Write} algorithm as $ER$ is at least $R_2$ (dup res with $Y_1$ already happened) or 2) $R_w \neq R_d$ which means there was a rebalance after dup res but that would have made the $PR$ at least $R_2$ - this is a contradiction to the assumption that $PR$ was $R_1$ at $X_1$ at some point of time (for $C_1$ to have RR of $R_1$).

    So $R_d \geq R_1$ which leads us to the following cases:

    \begin{casesp}
        \item \textbf{$R_d = R_w \in \{R_1, R_2\}$}

        Since dup-res from $Y_1$ to $X_1$ succeeds, $Y_1$ is part of regime $R_d$ which implies it is part of $R_w$ (as they are equal). Since $Y_1$ is a cluster leader in regime $R_2$ that performs dup-res, it is the first node in the succession list (last bullet in Step 5 of Rebalance Algorithm). Therefore, it has to be a roster replica (every cluster in which partition is available has a roster replica). This, in turn, implies it has to be a cluster replica in $R_w$. It will receive the replica write for $C_1$. Contradiction to Assumption~\ref{assumption:tworeplicatedchildren}.

        \item \textbf{$R_d = R_1$, $R_w = R_2$}

        Since $Y_1$ is a cluster replica in regime $R_2$, it must be among the recipients of the replica write for $C_1$. Again, contradiction.
    \end{casesp}

\end{casesp}

\item \textbf{$Y_1$ does not perform dup-res for $KV$ before writing $C_2$.}

In this case, the condition in Line 9 of Algorithm~\ref{alg:clientwrite} evaluates to false. Since $C_2$ is the first version in regime $R_2$, there cannot already be a version with regime $R_2$, and so the only way Line 9 is skipped is if $Y_1$ is full.

By Lemma~\ref{lemma:fullseeseverything}, Y1 would have seen C1 if it was successfully replicated by the time it attempts to write C2 - a contradiction to Assumption~\ref{assumption:tworeplicatedchildren}. If C1 gets replicated successfully without Y1's knowledge after C2 was written, that is a violation of Lemma~\ref{lemma:fullseeseverything}.

\end{casesp}

In all cases, Assumption~\ref{assumption:tworeplicatedchildren} leads to a contradiction. Therefore, a record version cannot have two children that are both replicated.
\end{proof}

\begin{theorem}
\label{theorem:writeextendschain}
    All writes form a linear chain of versions, each write building on the previous version. 
\end{theorem}

\begin{theorem}
\label{theorem:linearizable}
    All reads by LARK are linearizable. 
\end{theorem}

\section{Analytical Availability Model}
\label{app:analytical}

We model per-partition unavailability under independent node failures with small per-node unavailability $u$ (e.g., $u\approx \lambda d$ for Poisson failures with rate $\lambda$ and mean downtime $d$; in our simulator with per-tick failure probability $p$ and deterministic recovery $r$ ticks, $u \approx p\,r$, see below).

\paragraph{LARK.}
With replication factor $RF=f{+}1$, LARK becomes unavailable only if \emph{all} $RF$ roster replicas fail (and the database simultaneously loses majority, a higher-order event negligible at small $u$). The leading-order term is
\begin{equation}
  \Pr[\text{unavail}_\text{LARK}] \approx u^{\,f+1}.
\end{equation}

\paragraph{Raft (fixed $2f{+}1$-replica majority).}
A partition is unavailable if at least $f{+}1$ of its $2f{+}1$ fixed replicas fail:
\begin{equation}
  \Pr[\text{unavail}_\text{Raft}]
  = \sum_{k=f+1}^{2f+1} \binom{2f+1}{k} u^k (1-u)^{2f+1-k}
  \approx \binom{2f+1}{f+1} \, u^{\,f+1},
\end{equation}
approximating by the first term ($k{=}f{+}1$) for small $u$.

\paragraph{Improvement factor.}
The Raft-to-LARK ratio simplifies to the combinatorial multiplier:
\begin{equation}
\label{eqn:lark-raft-unavail-ratio}
  \frac{\Pr[\text{unavail}_\text{Raft}]}{\Pr[\text{unavail}_\text{LARK}]}
  \approx \binom{2f+1}{f+1}
  =
  \begin{cases}
    3 & f{=}1,\\
    10 & f{=}2,\\
    35 & f{=}3~.
  \end{cases}
\end{equation}

\paragraph{Mapping simulator $p$ to $u$.}
With per-tick failure probability $p$ and fixed downtime $r$, an alternating-renewal argument yields
\[
u \;=\; \frac{p\,r}{1+p\,r} \;\approx\; p\,r \quad (p\,r \ll 1).
\]
Substituting gives absolute unavailability and shows that increasing $r$ scales both protocols similarly, leaving the ratio unchanged.

\end{document}